\documentclass{article}
\usepackage[main,final]{neurips_2026}

\usepackage[utf8]{inputenc}
\usepackage{amsmath,amssymb,amsfonts,amsthm,mathtools,bm}
\usepackage{booktabs}
\usepackage{tikz}
\usepackage{graphicx}
\usepackage{dblfloatfix}
\usepackage{authblk}
\usepackage[algo2e]{algorithm2e}
\usepackage{algorithm}
\usepackage{algorithmic}
\usepackage{natbib}
\renewcommand{\bibsection}{\subsubsection*{References}}
\usepackage{enumitem}
\setlistdepth{9}
\setlist[itemize,1]{label=$\bullet$}
\setlist[itemize,2]{label=$+$}
\setlist[itemize,3]{label=$--$}
\setlist[itemize,4]{label=$\star$}
\setlist[itemize,5]{label=$\circ$}
\setlist[itemize,6]{label=$\bullet$}
\setlist[itemize,7]{label=$\bullet$}
\setlist[itemize,8]{label=$\bullet$}
\setlist[itemize,9]{label=$\bullet$}
\renewlist{itemize}{itemize}{9}
\usepackage{xcolor}
\usepackage{hyperref}
\hypersetup{colorlinks=true, linkcolor=blue, citecolor=blue, urlcolor=blue}
\usepackage{comment}
\usepackage{times}
\IfFileExists{math_commands.tex}{

\usepackage{amsmath,amsfonts,bm}

\def\eqref#1{equation~\ref{#1}}

\def\1{\bm{1}}

\DeclareMathAlphabet{\mathsfit}{\encodingdefault}{\sfdefault}{m}{sl}
\SetMathAlphabet{\mathsfit}{bold}{\encodingdefault}{\sfdefault}{bx}{n}

\newcommand{\E}{\mathbb{E}}

}{}

\newenvironment{proofsketch}{%
  \proof}{\endproof}

\newcommand{\norm}[1]{\left\lVert#1\right\rVert}

\newcommand{\cind}{\mathrel{\perp\mspace{-9mu}\perp}}

\newtheorem{theorem}{Theorem}

\newtheorem{lemma}[theorem]{Lemma}
\newtheorem{remark}{Remark}
\newtheorem{assumption}{Assumption}

\newtheorem{proposition}{Proposition}
\newtheorem{definition}{Definition}[section]

\SetKw{Continue}{continue}
\SetKw{Break}{break}

\title{Causal Discovery via Transformed Low-Rank Quantile Surfaces}

\author{%
\textbf{Ryo Kamimura}$^{1,3}$
\qquad
\textbf{Thong Pham}$^{2,1,3}$\thanks{Corresponding author.}
\\
$^{1}$Shiga University
\quad
$^{2}$The University of Osaka
\quad
$^{3}$RIKEN AIP
\\
\texttt{s6025124@st.shiga-u.ac.jp}
\quad
\texttt{thong-pham@ds.sanken.osaka-u.ac.jp}
}

\date{}

\begin{document}

\maketitle

\begin{abstract}
We propose Low-Rank Quantile Surfaces (LRQS), a bivariate causal model in which, in the causal direction, an unknown monotone transformation of the conditional quantile surface admits a low-rank functional decomposition. LRQS subsumes location-scale noise models and post-nonlinear heteroscedastic noise models, while allowing multiple quantile bases to represent changes beyond location-scale effects. We prove generic identifiability of LRQS: the transformed quantile surface is low rank in the causal direction, whereas reverse representability under the corresponding constraints occurs only for exceptional, fine-tuned cause marginals. We provide a nonparametric causal score using a fitting procedure that alternates between rank-constrained approximation of discretized quantile surfaces and isotonic estimation of the unknown monotone transformation. Experiments on synthetic mechanisms with higher-rank distributional shape variation and strong nonlinear distortions, together with standard bivariate benchmarks, show that LRQS is especially effective when conditional distributional shape or observation distortion goes beyond existing location-scale assumptions\footnote{Code: \url{https://github.com/Pham-Lab/lrqs}.}.
\end{abstract}

\section{Introduction}

Inferring causal direction from observational data requires asymmetry: the
conditional distribution in the causal direction should admit a simpler
description than the one in the anticausal direction. Classical approaches
instantiate this principle through structural restrictions such as additive noise
models (ANMs)~\citep{Hoyer09NIPS}, post-nonlinear models (PNLs)~\citep{Zhang09UAI}, and location-scale noise models (LSNMs)~\citep{Immer_2023}. These models impose location or location-scale structure, either directly or after an invertible transformation. They are identifiable because their assumptions are generally not preserved under
reversal, but their expressiveness is limited: in the latent scale, real conditional distributions may
vary not only in location and scale, but also in skewness, tail behavior, and
other shape features.

Quantile-based causal discovery provides a natural way to model distributional
asymmetry beyond conditional means. A recent line of work based on quantile
partial effects (QPE) assumes that the derivative of the conditional quantile
surface with respect to the conditioning variable lies in a finite span of known
basis functions~\citep{QPE_2026causal}. This is an expressive observational
restriction, but it is imposed on the quantile slope field rather than on the
quantile surface itself. Moreover, in their identifiability theory, the finite basis is fixed in advance and is
not designed to absorb an unknown monotone observation transformation such as those in PNLs.

We propose \emph{Low-Rank Quantile Surfaces} (LRQS), a bivariate causal model
that addresses this limitation. Let \(Q_{Y\mid X=x}(u)\) denote the conditional
quantile function of $Y$ given  $X=x$. LRQS assumes that, in the causal direction,
there exists an unknown increasing transformation \(h=g^{-1}\) such that
\[
    h\!\left(Q_{Y\mid X=x}(u)\right)
    =
    a(x) + \sum_{k=1}^{K} b_k(x)q_k(u).
\]
Thus the observed quantile surface need not be low rank; instead, it becomes
low rank after an unknown monotone unwarping. The case \(K=1\) recovers a
post-nonlinear heteroscedastic noise model (PNL-HNM), while larger \(K\) captures richer
changes in conditional shape beyond location-scale variation in the latent scale.

\begin{figure}[t]
  \centering
  \includegraphics[width=\linewidth]{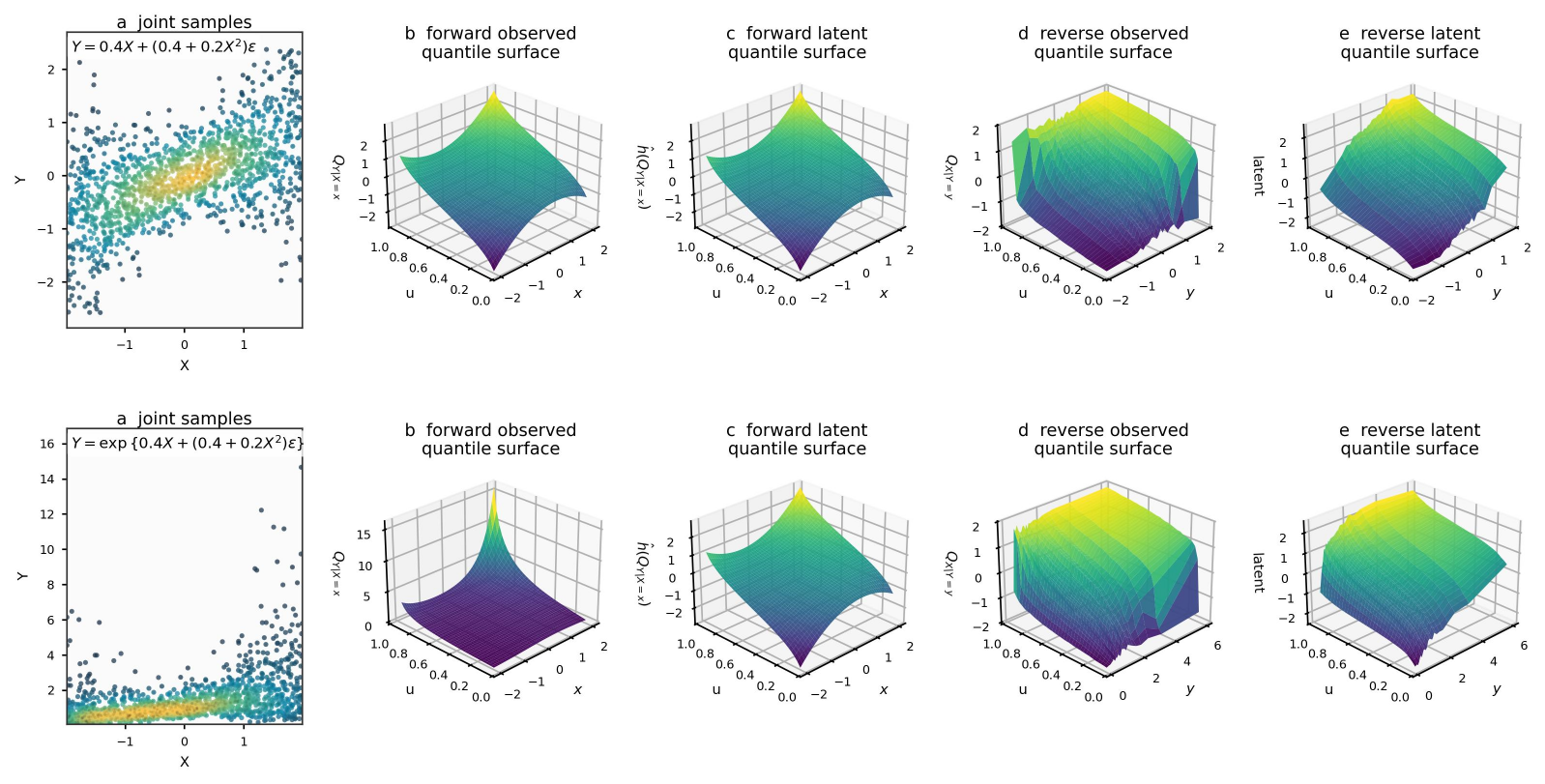}
  \caption{
  Observed and latent quantile surfaces. First row: an LSNM 
  \(Y=0.4X+(0.4+0.2X^2)\varepsilon\). Second row: a PNL-HNM
  \(Y=\exp\{0.4X+(0.4+0.2X^2)\varepsilon\}\).
  (a) Joint samples of \(X\) and \(Y\);
  (b) forward observed quantile surface \(Q_{Y\mid X=x}(u)\);
  (c) forward latent quantile surface \(\hat h(Q_{Y\mid X=x}(u))\);
  (d) reverse observed quantile surface \(Q_{X\mid Y=y}(u)\);
  (e) reverse latent quantile surface.
  In the causal direction, the observed surface becomes low rank after a monotone unwarping.
  }
  \label{fig:figure1}
\end{figure}

We prove that this transformed low-rank structure is generically identifiable.
In the causal direction, the transformed conditional quantile surface is low rank by construction. For every fixed finite number \(L\) of reverse non-intercept components and each prescribed reverse quantile basis, the compatible cause log-densities have local restrictions belonging to a finite-dimensional family, and thus are exceptional. A corresponding result with a free reverse basis holds for \(L=1\). Neither result requires \(L=K\). 

Our theory leads to the following nonparametric causal score. We estimate conditional
quantile matrices in both directions and fit the LRQS structure by alternating
between rank-constrained approximation and isotonic estimation of the unknown
monotone transformation. The direction with the smaller reconstruction error is
selected as causal. Figure~\ref{fig:figure1} gives a visual summary of the LRQS method. 
Our contributions are:
\begin{itemize}
    \item We introduce LRQS, a transformed low-rank conditional quantile model
    that extends location-scale and post-nonlinear heteroscedastic noise models.

    \item We establish generic identifiability by characterizing reverse-compatible cause marginals through finite-dimensional local restrictions on their log-densities. This gives, to our knowledge, the first
generic identifiability result for the nondegenerate bivariate PNL-HNM class
with unknown transformation and reverse noise distribution.

    \item We provide a practical causal discovery algorithm based on low-rank
    quantile-surface fitting and isotonic estimation of the unknown transformation.

    \item We show through experiments that LRQS is particularly effective on mechanisms with higher-rank conditional shape variation and strong nonlinear observation distortions, while remaining competitive on standard bivariate benchmarks.
\end{itemize}

\section{Related work}
Quantile-based causal discovery exploits asymmetries beyond the conditional
mean. Bivariate quantile causal discovery uses multiple conditional quantile
levels and an independence-of-mechanisms description-length principle
\citep{bQCD_2020}. 

The closely related quantile partial
effect (QPE) framework assumes that the derivative of the conditional quantile
surface with respect to the conditioning variable, i.e.
$\psi(x,y)=Q_x(x,F_{Y\mid X=x}(y))$, lies in a fixed finite span in $x$ and $y$ coordinates~\citep{QPE_2026causal}. LRQS differs by imposing low rank on the transformed
conditional quantile surface in $x$ and $u$ coordinates: both the quantile bases and the monotone
unwarping are learned rather than fixed in advance. The QPE and LRQS model classes overlap nontrivially, and neither
theory subsumes the other:
(a) ANMs and LSNMs belong to both frameworks;
(b) LRQS with an unknown distortion, such as PNL or PNL-HNM,
is not contained in QPE with any basis fixed in advance; and
(c) conversely, there are QPE models that cannot be expressed
as a transformed finite-rank representation for any fixed $K$
in the LRQS framework.
 
 Recent alternatives include
optimal-transport and velocity-field criteria, such as DIVOT and causal
velocity models, which characterize causal asymmetry through transport dynamics
or score/velocity equations \citep{tu2022optimal,xi2025distinguishing}, as well
as neural generative or likelihood-based bivariate methods
\citep{goudet2018learning,Immer_2023}. These methods are
complementary, but LRQS is targeted at an interpretable nonparametric
quantile-surface score with explicit transformed low-rank identifiability
guarantees. 

LRQS also connects to single-index quantile regression, transformation models, and random-coefficient quantile models in econometrics. See Appendix~\ref{appendix:related_works} for more detailed discussions.

\section{The model}
In the introduction we motivated LRQS as a transformed low-rank structure in the
conditional quantile surface. We now give the formal model definition:
\begin{equation}
Y = g\left(a(X) + \sum_{k=1}^{K} b_k(X) q_k(U)\right),
\label{eq:low_rank_quantile}
\end{equation}
where $U\sim\mathrm{Unif}(0,1)$, $U\cind X$, and $g$ is a
continuous, strictly increasing function with inverse
$h\coloneq g^{-1}$.

For a fixed $K$, the transformation $g$, the basis $\{q_k\}_{k=1}^{K}$, the shift function $a(x)$, and the coefficient functions $b_k(x)$ are \emph{unknown}.

\textbf{LRQS generalizes post-nonlinear heteroscedastic noise models.} When each $q_k$ is a continuous, increasing function, $q_k$ is the inverse CDF of some random variable $\varepsilon_k$, and thus
$
Y = g\left(a(X) + \sum_{k=1}^{K} b_k(X) \varepsilon_k\right)
$
 for some generally dependent noises $\varepsilon_k = q_k(U)$. When $K = 1$, we get $Y = g(a(X) + b(X)\varepsilon)$. This class is sometimes called a post-nonlinear heteroscedastic noise model (PNL-HNM)~\citep{QPE_2026causal}. It contains the LSNM~\citep{Immer_2023} when $g(z) = z$.

\textbf{Beyond location and scale.}
Multiple quantile bases $\{q_k\}_{k=1}^{K}$ allow the latent
conditional distributions to vary in shape, including skewness,
kurtosis, and tail behavior, rather than only in location and scale.
See Appendix~\ref{appendix:beyond_pnl_hnm} for a discussion.

The following assumption makes \(u\) the conditional quantile level: increasing the latent noise rank increases the response, so the structural function can be read directly as a conditional quantile function.
\begin{assumption}\label{assumption:monotone}
For every fixed $x$, the function
$
z(x,u) \coloneq a(x) + \sum_{k=1}^{K} b_k(x) q_k(u)
$
is continuous and strictly increasing in $u \in (0,1)$.
\end{assumption}
This assumption holds, for example, if each $b_k(x)$ is positive and each basis function $q_k$ is continuous and strictly increasing.

The following proposition, whose proof is in Appendix~\ref{appendix:proof}, shows the low-rank structure of the transformed quantile surface.
\begin{proposition}\label{proposition:factorization}
Denote by $Q_{Y\mid X=x}(u)$ the conditional quantile function of $Y\mid X=x$. Under Assumption~\ref{assumption:monotone}, the conditional quantile surface is
\begin{equation}
Q(x,u) \coloneq Q_{Y\mid X=x}(u) = g\left(a(x) + \sum_{k=1}^{K} b_k(x) q_k(u)\right).
\label{eq:conditional_quantile_surface}
\end{equation}
Equivalently,
$
h\bigl(Q(x,u)\bigr) = a(x) + \sum_{k=1}^{K} b_k(x) q_k(u).
$
\end{proposition}

A surface $H(x,u)$ has rank at most $d$ if it can be written as
$H(x,u)=\sum_{j=1}^{d}A_j(x)B_j(u)$.
Proposition~\ref{proposition:factorization} therefore implies that
$h(Q(x,u))$ has rank at most $K+1$, with the additional component
corresponding to the intercept $a(x)$.
Evaluating this surface on a grid gives a matrix with the same
rank upper bound, which motivates our fitting procedure.
The observed surface $Q(x,u)$, by contrast, need not have finite rank.

\begin{remark} $Y$ as defined in Eq.~(\ref{eq:low_rank_quantile}) is still a valid random variable without Assumption~\ref{assumption:monotone}, for example, when there exists some $x$ such that $z(x,u)$ is not increasing in $u$. However, in that case the induced conditional quantile surface need not retain the transformed low-rank structure above.
\end{remark}

\section{Identifiability}
Fix the forward conditional density $r(y\mid x)=p_{Y\mid X}(y\mid x)$
and write $\xi(x)=\log p_X(x)$.
Changing $\xi$ preserves the forward LRQS mechanism but changes the
reverse conditional distribution through Bayes' rule.
We study which cause marginals also make the reverse quantile surface
$Q^\leftarrow(y,u):=Q_{X\mid Y=y}(u)$ belong to a specified LRQS class.
We use $K$ and $L$ for the forward and reverse numbers of
non-intercept components, respectively, with no relation imposed
between them.

The following definition formalizes the set of such tuned marginals:
\begin{definition}[Reverse-compatible marginal set]
Fix the forward conditional density $r(y\mid x)$.
For a backward model class $\mathcal M$, define
\begin{equation}
\mathcal B(r;\mathcal M)
:=
\Big\{
\xi:\ p_{X,Y}(x,y)=e^{\xi(x)}r(y\mid x)\ \text{and the induced }Q^{\leftarrow}(\cdot,\cdot)\in\mathcal M
\Big\}.
\end{equation}
\end{definition}

\paragraph{Notion of genericity.}
For an interval $I$, write $\xi|_I$ for the restriction of $\xi$ to $I$.
Our results place such restrictions of reverse-compatible
log-densities in common finite-parameter families.
The interval may depend on $\xi$, but the family depends only on
the fixed forward mechanism and the specified reverse class.
We use generic identifiability in this local finite-dimensional sense.
Here, ``local'' describes the restriction on the exceptional
marginals.

For each reverse-compatible distribution under consideration,
the following two assumptions are required at the same interior
point $(x_0,y_0)$.
This point may vary between distributions.
Write $u_0=F_{X\mid Y=y_0}(x_0)$ and
$\alpha_{y_0}(x)
=\left.\partial_y\log r(y\mid x)\right|_{y=y_0}$,
with primes on $\alpha_{y_0}$ denoting differentiation in $x$.

\begin{assumption}[Local score variation]
\label{assumption:anchor}
The conditional $y$-score has a nonzero derivative in $x$ at $x_0$:
$\alpha_{y_0}'(x_0)\ne0$.
\end{assumption}

Thus $\alpha_{y_0}$ is locally invertible near $x_0$, allowing us to
recover a reverse quantile curve from the Bayes identity used in
the proofs.

\begin{assumption}[Local regularity]
\label{assumption:smooth}
In neighborhoods of the corresponding arguments, the densities, conditional quantile
maps, and functions defining the reverse representation are $C^3$,
and the relevant densities are strictly positive.
\end{assumption}

These conditions ensure that the logarithmic derivatives and inverse conditional quantiles used below are well-defined locally, and that \(Q_u^\leftarrow(y,u)=1/p_{X\mid Y=y}(Q^\leftarrow(y,u))>0\) locally.

The general idea in both Theorems~\ref{theorem:general_K_fixed_quantile} and~\ref{theorem:K__1_free_quantile} below is that, at a fixed conditioning value \(y_0\), a reverse quantile curve determines the corresponding conditional density of \(X\). Bayes’ rule then recovers the shape of the cause density from this conditional density and the fixed forward mechanism. Thus, restricting one reverse quantile curve at \(y_0\) already restricts the cause density on an interval.

\paragraph{Main message.}
We establish these local restrictions for any prescribed finite reverse
quantile basis and, when $L=1$, for an unrestricted reverse basis.
The reverse monotone transformation is unrestricted in both cases.

\subsection{General $L$, fixed basis}
We consider the following backward model class:
\begin{equation}
\label{eq:backward-fixed}
\mathcal M_{L,\text{fixed}}
=
\left\{
Q^\leftarrow:
Q^\leftarrow(y,u)
=\tilde g\left(
\tilde a(y)+\sum_{k=1}^{L}\tilde b_k(y)\tilde q_k(u)
\right)
\text{ for some }\tilde g,\tilde a,\{\tilde b_k\}_{k=1}^{L}
\right\}.
\end{equation}
The reverse basis
$\tilde q_1,\ldots,\tilde q_L:(0,1)\to\mathbb R$
is prescribed, while $\tilde g,\tilde a,\tilde b_1,\ldots,\tilde b_L$
remain unrestricted, subject to the model assumptions.

\textbf{The key object in the analysis} is the backward \emph{drift ratio} 
$
R(y,u) = \partial_{y}Q^{\leftarrow}(y,u)/\partial_u Q^{\leftarrow}(y,u) 
$, 
which measures how the $u$-th backward quantile of $X\mid Y = y$ moves as the conditioning value $y$ changes, normalized by $Q_u^{\leftarrow}$. When $Q^{\leftarrow} \in \mathcal{M}_{L,\text{fixed}}$, the outer monotone distortion cancels from this ratio, so $R$ exposes constraints imposed by the inner low-rank structure:
\begin{equation}
R(y,u)
= \dfrac{\tilde{g}'(z^{\leftarrow}(y,u))\partial_yz^{\leftarrow}(y,u)}{\tilde{g}'(z^{\leftarrow}(y,u))\partial_u z^{\leftarrow}(y,u)} = 
\frac{\tilde a'(y)+\sum_{k=1}^{L}\tilde b_k'(y)\tilde q_k(u)}
{\sum_{k=1}^{L}\tilde b_k(y)\tilde q_k'(u)},\label{eq:backward_drift_ratio}
\end{equation}
where $z^{\leftarrow}(y,u) = 
\tilde a(y)+\sum_{k=1}^{L}\tilde b_k(y)\,\tilde q_k(u)$.

\begin{theorem}[Prescribed reverse quantile basis]
\label{theorem:general_K_fixed_quantile}
Fix the forward conditional density $r$, a finite integer $L\ge1$,
and a reverse quantile basis $\{\tilde q_k\}_{k=1}^L$. Every reverse-compatible cause
log-density $\xi\in\mathcal B(r;\mathcal M_{L,\text{fixed}})$ satisfying both Assumptions~\ref{assumption:anchor}
and~\ref{assumption:smooth} at some interior point
agrees, on some nonempty open interval, with a member of a common
family parameterized by at most $2L+3$ real numbers.
This family depends only on $r$ and the prescribed reverse basis.
\end{theorem}

\begin{proofsketch}

Fix an interior point as in the assumptions and write \(x(u)=Q^\leftarrow(y_0,u)\). At the fixed conditioning value \(y_0\), Eq.~(\ref{eq:backward_drift_ratio}) expresses \(R(y_0,\cdot)\) using \(2L+1\) coefficient values. A common nonzero scaling leaves the ratio unchanged, so at most \(2L\) parameters are needed.

Bayes’ rule gives \(R_u(y_0,u)=c_0-\alpha_{y_0}(x(u))\), where \(c_0=(\log p_Y)'(y_0)\). By Assumption~\ref{assumption:anchor}, \(\alpha_{y_0}\) is locally invertible. Thus those parameters, together with \(c_0\), determine the reverse quantile curve \(x(u)\) and its derivative.

A second application of Bayes’ rule gives \(\xi(x(u))=\ell_0-\log r(y_0\mid x(u))-\log x_u(u)\), where \(\ell_0=\log p_Y(y_0)\). Since \(x_u>0\), this determines \(\xi\) on the corresponding \(x\)-interval. Counting the \(2L\) coefficient parameters, \(c_0,\ell_0\), and \(y_0\) gives the upper bound \(2L+3\). The full proof is given in Appendix~\ref{appendix:proof_theorem_1}.
\end{proofsketch}

\begin{remark}
Theorem~\ref{theorem:general_K_fixed_quantile} allows an arbitrary
unknown reverse distortion $\tilde g$, which is not covered by
the prescribed-basis QPE theory~\citep{QPE_2026causal}.
\end{remark}

\subsection{$L = 1$, free basis}
We next consider reverse representations with one non-intercept
component, allowing their quantile basis to vary:
\begin{equation}
\label{eq:backward-free}
\mathcal M_{1,\text{free}}
=
\left\{
Q^\leftarrow:
Q^\leftarrow(y,u)
=\tilde g\bigl(\tilde a(y)+\tilde b(y)\tilde q(u)\bigr)
\text{ for some }\tilde g,\tilde a,\tilde b,\tilde q
\right\}.
\end{equation}
The proof strategy of
Theorem~\ref{theorem:general_K_fixed_quantile}
does not apply directly because $R(y_0,\cdot)$ need not belong to
a finite-parameter family when $\tilde q$ is unrestricted.

\textbf{The key identity} is given in the following proposition:
\begin{proposition}
When $L = 1$,
\begin{equation}
\partial_uR +T(u)R=\rho(y),\label{eq:ODE_R}
\end{equation}
where $T(u) = \tilde{q}''(u)/\tilde{q}'(u)$ and $\rho(y) = \tilde{b}'(y)/\tilde{b}(y)$.
\end{proposition}
\begin{proof}
Since $R =\left(\tilde a'(y)+ \tilde b'(y)\tilde q(u)\right)/\left(
\tilde b(y)\tilde q'(u)\right)$, 
$
\partial_u R = -\tilde{q}''(u)/\tilde{q}'(u)R + \tilde{b}'(y)/\tilde{b}(y)
$.
\end{proof}

\begin{theorem}[Free reverse quantile basis for $L=1$]
\label{theorem:K__1_free_quantile}
Fix the forward conditional density $r$. Every reverse-compatible cause
log-density $\xi\in\mathcal B(r;\mathcal M_{1,\text{free}})$ satisfying both  Assumptions~\ref{assumption:anchor}
and~\ref{assumption:smooth} at some interior point 
agrees, on some nonempty open interval, with a member of a common
family parameterized by at most nine real numbers.
This family depends only on $r$; no reverse quantile basis is
prescribed.
\end{theorem}

\begin{proofsketch} The key observation is that the unknown reverse basis $\tilde{q}$ enters Eq.~(\ref{eq:ODE_R}) only through \(T(u)\), which does not depend on \(y\). Fix $y_0$ as in the assumptions
and set $x(u)=Q^\leftarrow(y_0,u)$, $R(u)=R(y_0,u)$, and
$D(u)=R_y(y_0,u)$. Equation~(\ref{eq:ODE_R}) and its $y$-derivative give
$R_u+TR=\rho_0$ and $D_u+TD=\rho_1$, where
$\rho_0=\rho(y_0)$ and $\rho_1=\rho'(y_0)$.
Multiplying the second identity by $R$ and subtracting $D$ times
the first eliminates $T$:
\[
R D_u-D R_u=\rho_1R-\rho_0D.
\]
Bayes' rule expresses $R_u$ in terms of $x$.
Differentiating the Bayes identity in $y$ before restricting to $y_0$
also expresses $D_u$ in terms of $x$, $R$, and $x_u$.
Substitution into the identity above determines $x_u$ on a suitable
subinterval, giving a first-order system for $(x,R,D)$ with no
unknown quantile-basis function remaining.

Under Assumptions~\ref{assumption:anchor} and~\ref{assumption:smooth}, local uniqueness determines the solution
from three initial values and the four constants
$c_0=(\log p_Y)'(y_0)$, $c_1=(\log p_Y)''(y_0)$,
$\rho_0$, and $\rho_1$.
Bayes' rule then recovers $\xi$ on the corresponding $x$-interval
after specifying $\ell_0=\log p_Y(y_0)$.
Neither the system nor the reconstruction depends explicitly on $u$,
so the origin of the local $u$-parameter need not be counted.
Including $\ell_0$ and $y_0$ gives at most $3+4+2=9$ real parameters.
The full proof is given in Appendix~\ref{appendix:proof_theorem_2}.
\end{proofsketch}

\begin{remark}
The elimination in the proof sketch relies on the separation of
$u$ and $y$ into $T(u)$ and $\rho(y)$ in Eq.~(\ref{eq:ODE_R}).
Our argument does not provide a corresponding elimination for
multiple reverse components, so the free-basis result is limited
to $L=1$.
\end{remark}

\begin{remark}
To our knowledge, when both candidate directions are modeled as PNL-HNM ($K=L=1$), Theorem~\ref{theorem:K__1_free_quantile} gives the first generic identifiability result for the nondegenerate PNL-HNM class $Y = g(a(X) + b(X)\varepsilon)$. Earlier identifiability results cover important special cases, including ANMs ($g(z)=z$, $b\equiv 1$)~\citep{Hoyer09NIPS}, PNL models ($b \equiv  1$)~\citep{Zhang09UAI}, and LSNMs ($g(z) = z$)~\citep{Immer_2023}.
\end{remark}

\paragraph{Scope of the assumptions.}
Assumptions~\ref{assumption:anchor} and~\ref{assumption:smooth}
require a common interior point where $\partial_x\partial_y\log p_{X,Y}(x_0,y_0)
=\alpha_{y_0}'(x_0)\ne0$.
This condition excludes uniform-noise location-scale representations in either
direction, even with varying scale and an unknown monotone
transformation, as well as constant-scale exponential and Laplace noise models.
In the ANM special case, the same nonvanishing condition appears
in classical differential-equation arguments for generic
identifiability~\citep{Hoyer09NIPS,Peters14JMLR}. Our assumptions also exclude pure-scale power-law noise models. These restrictions concern the scope of the theorems and do not imply that every excluded mechanism is nonidentifiable. See Appendix~\ref{appendix:assumption_scope} for details.

\section{Causal scoring by transformed quantile-surface fitting}
For a chosen number $K_{\mathrm{fit}}$ of non-intercept components, we consider the following causal score:
\begin{equation}
S_{X\rightarrow Y} = \inf_{\substack{a,b_k,q_k,\\g: \text{ increasing},\\ z(x,\cdot): \text{ increasing}}} \E_X\int_0^{1}\left(Q(X,u) - g\left(a(X) + \sum_{k=1}^{K_{\mathrm{fit}}} b_k(X) q_k(u)\right)\right)^2 \, du,
\label{eq:population_score}
\end{equation}
where $z(x,u) \coloneq a(x) + \sum_{k=1}^{K_{\mathrm{fit}}} b_k(x) q_k(u)$.

Suppose the forward conditional quantile surface belongs to the LRQS class with $K \le K_{\mathrm{fit}}$. Then the causal population score satisfies $S_{X\rightarrow Y} =0$.

In the reverse direction, the same fitted component count corresponds to \(L=K_{\mathrm{fit}}\). Theorem~\ref{theorem:general_K_fixed_quantile} constrains exact reverse representations with a prescribed basis for every fixed finite \(L\), whereas Theorem~\ref{theorem:K__1_free_quantile} allows a free reverse basis when \(L=1\). For a population score whose reverse fitting class is covered by the relevant theorem, assume additionally that zero score implies exact membership in that class. Then the reverse score is positive outside the corresponding reverse-compatible marginal set. This set satisfies the finite-dimensional local restrictions established above, so the reverse score is positive for generic cause marginals.

The population score above, as well as the theoretical results in the
previous section, is defined for continuous conditional quantile surfaces.
The discretization described below is introduced only for finite-sample
estimation.

For a candidate direction $X \rightarrow Y$, we first sort the observations
by the conditioning variable $X$ and partition the sorted observations into
$G$ bins with sizes as equal as possible. Thus, our implementation uses
approximately equal-count binning. For the $j$-th bin $I_j$, the observed conditional quantile matrix is defined by
$
[\mathbf{Q}_{\mathrm{obs}}]_{j,l}
=
\widehat{Q}_{Y \mid X \in I_j}(u_l)$, for $
j=1,\ldots,G$, and $l=1,\ldots,B$. Given the $G \times B$ observed quantile matrix, we estimate the
population causal score by alternating between a low-rank approximation of
the latent surface and estimation of the monotone transformation. The shift
function $a(x)$ becomes a length-$G$ vector $\mathbf{a}$, while
$
a(x)+\sum_{k=1}^{K_{\mathrm{fit}}} b_k(x)q_k(u)
$
becomes a $G \times B$ matrix $\mathbf{Z}$. Further implementation details are
provided in Appendix~\ref{appendix:experiment_details}. The algorithm \texttt{LowRank} computes an empirical approximation to $S_{X\rightarrow Y}$, and the algorithm \texttt{Bivariate-LRQS} chooses the direction with the smaller causal score.

\begin{algorithm}[!h]
\caption{\texttt{LowRank}$(\mathbf{Q},K_{\mathrm{fit}})$}\label{alg:LowRank}
\KwData{$G\times B$ matrix $\mathbf{Q}$, number of non-intercept components $K_{\mathrm{fit}}$, inner and outer loop iterations $T_{\mathrm{in}}$ and $T_{\mathrm{out}}$}
\KwResult{Approximation score $s$}

Initialize $\mathbf{Z} = \mathbf{Q} + \mathbf{\Sigma}$, where $\mathbf{\Sigma}$ is a random Gaussian perturbation matrix.

Repeat for $t = 1, \dots, T_{\mathrm{out}}$:
\begin{itemize}
    \item Find an increasing function $g$ such that $g(\mathbf{Z})$ best approximates $\mathbf{Q}$ under the Frobenius norm: flatten the matrices $\mathbf{Q}$ and $\mathbf{Z}$, sort the pairs $(z,q)$ by the $z$-values, and perform isotonic regression.

    \item Update $\mathbf{Z}$ by an approximate inverse of $g$:
    $
    \mathbf{Z}[j,l] \leftarrow \hat{h}(\mathbf{Q}[j,l])
    $, 
    where $\hat{h}$ denotes the inverse mapping approximated by linear interpolation from the isotonic fit.

    \item Repeat this inner loop for $\tau = 1, \dots, T_{\mathrm{in}}$ to project $\mathbf{Z}$ onto the low-rank structure:
    \begin{itemize}
        \item $\mathbf{a} \leftarrow \mathrm{RowMeans}(\mathbf{Z})$ and subtract the intercept function: $\mathbf{Z}_c \leftarrow \mathbf{Z} - \mathbf{a}\mathbf{1}^{\intercal}$.
        \item Use truncated SVD to update $\mathbf{Z}_c$: $\mathbf{Z}_c \leftarrow \mathrm{SVD}_{K_{\mathrm{fit}}}(\mathbf{Z}_c)$, and rebuild $\mathbf{Z} \leftarrow \mathbf{Z}_c + \mathbf{a}\mathbf{1}^{\intercal}$.
        \item Make each row of $\mathbf{Z}$ an increasing sequence by isotonic projection.
        \item Rescale $\mathbf{Z}$ so that its overall mean is $0$ and variance is $1$.
    \end{itemize}
\end{itemize}

Estimate the final increasing function $g$ from the flattened pairs $(z,q)$ of $\mathbf{Z}$ and $\mathbf{Q}$ by isotonic regression, and calculate the score $s = \norm{\mathbf{Q} - g(\mathbf{Z})}_{F}$.

Return the score $s$.
\end{algorithm}

\begin{algorithm}[!h]
\caption{\texttt{Bivariate-LRQS}}\label{alg:Bivariate_LRQS}
\KwData{Data matrix of $(X,Y)$, number of non-intercept components $K_{\mathrm{fit}}$, number of initializations $m$}

Standardize the data.

Build the conditional quantile matrix $\mathbf{Q}_{X\rightarrow Y}$ of $Y\mid X=x$ and compute the minimum score over $m$ independent runs to reduce sensitivity to local optima: $
s_{X\rightarrow Y} = \min_{1 \le i \le m} \texttt{LowRank}(\mathbf{Q}_{X\rightarrow Y},K_{\mathrm{fit}})$.

Build the conditional quantile matrix $\mathbf{Q}_{Y\rightarrow X}$ of $X\mid Y=y$ and compute the minimum score over $m$ independent runs: $
s_{Y\rightarrow X} = \min_{1 \le i \le m} \texttt{LowRank}(\mathbf{Q}_{Y\rightarrow X},K_{\mathrm{fit}})$.

If $s_{X \rightarrow Y} < s_{Y \rightarrow X}$, output $X$ as the parent; otherwise, output $Y$ as the parent.
\end{algorithm}

\section{Experiments}
\label{sec:experiments}

\subsection{Experimental setup}

\textbf{Custom benchmarks:}
To clarify the theoretical limitations of existing methods and evaluate the
robustness of our proposed approaches, we generated custom benchmarks consisting
of two categories, as summarized in Table~\ref{tab:custom_benchmarks}.
For each benchmark variant, we generated 100 independent cause-effect pairs
with a sample size of $n=1000$.
The cause variable $x$ is sampled from a Gaussian distribution
$\mathcal{N}(0,2)$.
For each model, we employed three types of noise components $\varepsilon$,
all normalized to have mean $0$ and variance $1$:
Gaussian $\mathcal{N}(0,1)$,
uniform $\mathcal{U}(-\sqrt{3},\sqrt{3})$,
and beta, where $V\sim\mathrm{Beta}(2,2)$ is scaled as
$(V-0.5)\times\sqrt{20}$.

\textit{1. Structural complexity ($K_{\mathrm{gen}}>1$):}
To evaluate cases with complex causal mechanisms without nonlinear
observational distortion, the DGP is
$
y = a(x) + \sum_{k=1}^{K_{\mathrm{gen}}} b_k(x)q_k(u)
$, 
where $a \sim\mathcal{GP}$ and $b_k(x)=|f_k(x)|$ with
$f_k\sim\mathcal{GP}$.
The noise basis functions $q_k(u)$ are standardized odd polynomials.
We set $K_{\mathrm{gen}}\in\{2,3,4,5\}$, resulting in 12 variants
(4 component counts $\times$ 3 noise types). This benchmark isolates conditional shape variation without an observation-level transformation, allowing comparison of the methods in the absence of nonlinear measurement distortion.

\textit{2. Strong nonlinear distortions (PNL):}
The DGP is the PNL-HNM model $y=g(z)$,
where $z$ follows either an ANM,
$z=f(x)+\varepsilon$, or an LSNM,
$z=f(x)+\sigma(x)\varepsilon$.
We consider five transformations:
(a) identity, $g(z)=z$;
(b) cube, $g(z)=z^3$;
(c) sigmoid, $g(z)=1/(1+\exp(-z))$;
(d) exponential, $g(z)=\exp(z)$; and
(e) hyperbolic tangent, $g(z)=\tanh(z)$.
This category consists of 30 variants.

Table~\ref{tab:custom_benchmarks} displays representative results for the
Gaussian-noise settings, including the exponential and hyperbolic-tangent
distortions.
The reported Avg.~acc. and Time (s), however, are macro-averaged over all
42 custom benchmark variants (12 structural-complexity variants and 30 PNL
variants).
Complete results for the 42 variants are provided in
Appendix~\ref{appendix:experiment_full_result}.

\textbf{Existing benchmarks:}
We evaluate all 24 bivariate benchmark datasets used in
\citet{QPE_2026causal}.
Table~\ref{tab:main_results} reports 12 representative datasets:
(i) AN and LS from \citet{bQCD_2020};
(ii) SIM and SIM-c from \citet{mooij2016distinguishing};
(iii) Cha and Net from \citet{guyon2019cause};
(iv) Per and Sig from \citet{xi2025distinguishing};
(v) Qd-V and NN-V from \citet{QPE_2026causal};
(vi) Tue~\citep{mooij2016distinguishing}; and
(vii) D4-s1~\citep{marbach2009generating}.
Although Table~\ref{tab:main_results} displays these 12 datasets,
its Avg.~acc. and Time (s) are macro-averaged over all 24 benchmark datasets.

\textbf{Baselines and evaluation protocol:}
We compare LRQS against ANM~\citep{zheng2024causal},
DIVOT~\citep{tu2022optimal},
CVEL~\citep{xi2025distinguishing},
QPE-k, and three QPE-f variants, which are fixed, poly, and lowrank, from
\citet{QPE_2026causal}.
For Table~\ref{tab:custom_benchmarks}, we report results obtained by running
these baselines on the same custom benchmark datasets.
For the existing-benchmark comparison in Table~\ref{tab:main_results},
we rerun all methods rather than citing benchmark values from prior work.
No dataset-specific hyperparameter tuning is performed for this main
comparison; instead, a single default configuration for each method is used
across all 24 datasets.
Implementation and computational details are provided in Appendix~\ref{appendix:comp_resource_implementation_detail}.

\textbf{Proposed method variants:}
We evaluate two configurations of LRQS:
\textit{fix-g}, which fixes the transformation to the identity mapping
($g(z)=z$), and
\textit{est-g}, which jointly estimates the unknown increasing transformation
$g$ together with the low-rank quantile surface.
Unless otherwise stated, both variants use the default settings described in
Appendix~\ref{appendix:experiment_details}.

\subsection{Main results}

Tables~\ref{tab:custom_benchmarks} and \ref{tab:main_results} summarize the causal direction identification accuracy and average execution time for the custom and existing benchmarks, respectively.

\begin{table*}[ht]
\centering
\caption{
Accuracy comparison on custom benchmarks highlighting theoretical limitations
of existing methods. The displayed columns show representative Gaussian-noise
settings, while Avg.~acc. and Time (s) are averaged over all 42 custom benchmark
variants.
}
\label{tab:custom_benchmarks}
\resizebox{\textwidth}{!}{
\begin{tabular}{lcccccccccc}
\toprule
& \multicolumn{4}{c}{
\begin{tabular}[c]{@{}c@{}}
Structural complexity\\
($g(z)=z, K_{\mathrm{gen}}>1$)
\end{tabular}}
& \multicolumn{4}{c}{
\begin{tabular}[c]{@{}c@{}}
Strong nonlinear\\
distortions (PNL)
\end{tabular}}
& \multicolumn{1}{c}{Overall}
& \\
\cmidrule(lr){2-5}
\cmidrule(lr){6-9}
\cmidrule(lr){10-10}
Method
& $K_{\mathrm{gen}}=2$
& $K_{\mathrm{gen}}=3$
& $K_{\mathrm{gen}}=4$
& $K_{\mathrm{gen}}=5$
& AN (Exp)
& AN (Tanh)
& LS (Exp)
& LS (Tanh)
& \begin{tabular}[c]{@{}c@{}}Avg.~acc.\\(all 42)\end{tabular}
& Time (s) \\
\midrule
ANM
& 0.44 & 0.48 & 0.43 & 0.40
& 0.21 & 0.27 & 0.26 & 0.33
& 0.421 & 7.335 \\

DIVOT
& 0.04 & 0.00 & 0.00 & 0.00
& 0.00 & 0.31 & 0.00 & 0.08
& 0.217 & 1.261 \\

CVEL
& 0.59 & 0.91 & 0.98 & \textbf{1.00}
& \textbf{1.00} & 0.70 & 0.94 & 0.66
& 0.730 & 2.055 \\

QPE-k
& 0.36 & 0.13 & 0.09 & 0.11
& 0.16 & 0.10 & 0.32 & 0.45
& 0.561 & 0.062 \\

QPE-f (fixed)
& 0.34 & 0.28 & 0.40 & 0.37
& 0.68 & 0.13 & 0.90 & 0.47
& 0.443 & 34.536 \\

QPE-f (poly)
& 0.33 & 0.18 & 0.17 & 0.15
& 0.27 & 0.07 & 0.53 & 0.44
& 0.390 & 32.588 \\

QPE-f (lowrank)
& 0.36 & 0.37 & 0.63 & 0.64
& 0.61 & 0.31 & 0.85 & 0.54
& 0.535 & 37.966 \\

\midrule

LRQS (est-g)
& 0.94 & \textbf{0.97} & 0.96 & \textbf{1.00}
& \textbf{1.00} & \textbf{0.91} & \textbf{1.00} & \textbf{0.96}
& \textbf{0.952} & 2.222 \\

LRQS (fix-g)
& \textbf{0.95} & \textbf{0.97} & \textbf{0.99} & \textbf{1.00}
& \textbf{1.00} & 0.59 & \textbf{1.00} & 0.78
& 0.893 & \textbf{0.050} \\

\bottomrule
\end{tabular}
}
\end{table*}

\begin{table*}[ht]
\centering
\caption{
Accuracy and average computational time (seconds per pair) on 12 representative
bivariate benchmark datasets using default parameter settings.
Avg.~acc. and Time (s) are averaged over all 24 benchmark datasets.
}
\label{tab:main_results}
\resizebox{\textwidth}{!}{
\begin{tabular}{lcccccccccccccc}
\toprule
Method
& AN
& LS
& SIM
& SIM-c
& Cha
& Net
& Per
& Sig
& Qd-V
& NN-V
& Tue
& D4-s1
& \begin{tabular}[c]{@{}c@{}}Avg.~acc.\\(all 24)\end{tabular}
& Time (s) \\
\midrule

ANM
& \textbf{1.00}
& 0.42
& 0.74
& \textbf{0.79}
& \textbf{0.73}
& 0.73
& 0.63
& 0.28
& 0.82
& 0.64
& 0.55
& 0.58
& 0.609
& 22.354 \\

DIVOT
& \textbf{1.00}
& 0.72
& 0.73
& 0.70
& 0.52
& 0.80
& 0.90
& 0.61
& 0.37
& 0.40
& 0.46
& \textbf{0.75}
& 0.612
& 4.057 \\

CVEL
& 0.25
& 0.18
& 0.64
& 0.62
& 0.67
& 0.51
& \textbf{1.00}
& 0.69
& \textbf{0.83}
& \textbf{0.88}
& 0.34
& 0.42
& 0.602
& 2.561 \\

QPE-k
& 0.99
& \textbf{1.00}
& \textbf{0.83}
& \textbf{0.79}
& 0.60
& 0.89
& 0.77
& \textbf{0.89}
& 0.42
& 0.53
& 0.54
& 0.58
& 0.741
& 0.067 \\

QPE-f (fixed)
& \textbf{1.00}
& 0.98
& 0.74
& 0.70
& 0.54
& 0.89
& 0.95
& 0.47
& 0.76
& 0.75
& 0.61
& 0.54
& 0.739
& 35.502 \\

QPE-f (poly)
& 0.97
& 0.99
& 0.75
& 0.66
& 0.50
& \textbf{0.91}
& 0.96
& 0.54
& 0.64
& 0.72
& 0.61
& 0.54
& \textbf{0.744}
& 31.672 \\

QPE-f (lowrank)
& 0.51
& 0.52
& 0.74
& 0.66
& 0.55
& 0.63
& 0.83
& 0.73
& 0.73
& 0.77
& 0.48
& 0.33
& 0.631
& 3.862 \\

\midrule

LRQS (est-g)
& 0.98
& 0.96
& 0.61
& 0.52
& 0.61
& 0.71
& 0.64
& 0.54
& 0.57
& 0.61
& 0.69
& 0.50
& 0.654
& 2.707 \\

LRQS (fix-g)
& \textbf{1.00}
& \textbf{1.00}
& 0.69
& 0.77
& 0.70
& 0.85
& 0.74
& 0.52
& 0.73
& 0.79
& \textbf{0.78}
& 0.42
& 0.740
& \textbf{0.053} \\

\bottomrule
\end{tabular}
}
\end{table*}

\textbf{Lightweight.} Both fix-g and est-g variants are very fast compared to neural-based methods such as QPE-f.

\textbf{Robustness to structural complexity.}
As shown in Table~\ref{tab:custom_benchmarks}, QPE-k and the QPE-f variants struggle to achieve consistently high accuracy under multirank structural complexity ($K_{\mathrm{gen}} > 1$), whereas CVEL performs strongly for larger $K_{\mathrm{gen}}$. In contrast, both of our proposed methods (fix-g and est-g) maintain near-perfect accuracy across the evaluated values of $K_{\mathrm{gen}}$. Notably, fix-g achieves this robustness while operating at a speed comparable to the fastest baselines.

\textbf{Fitted rank and effective complexity.} The fitted number of components need not match the generating number. In the Gaussian benchmark with $K_{\mathrm{gen}}=5$, both variants achieve 99–100\% accuracy across $K_{\mathrm{fit}}\in\{2,3,4\}$. Spectral analysis shows that the first two components capture over 99\% of the forward row-centered quantile matrix’s energy in the median case, helping to explain the strong performance obtained with the evaluated setting $K_{\mathrm{fit}}=2$. Increasing the fitted rank further can reduce accuracy, suggesting that $K_{\mathrm{fit}}$ should be viewed as a regularization parameter (Appendix~\ref{appendix:rank_sensitivity}).

\textbf{Resistance to strong nonlinear distortions.}
Under strong nonlinear distortions such as \textit{exp} and \textit{tanh}, several baseline methods experience substantial degradation in accuracy. While our fix-g also struggles with extreme saturation (e.g., \textit{tanh}), est-g successfully estimates and unwarps these distortions, sustaining high accuracy. Across all 42 custom benchmark variants, est-g achieves the highest average accuracy, further supporting the benefit of learning the monotone unwarping under strong observation-level distortions.

\textbf{Performance on existing benchmarks.}
On the standard datasets (Table~\ref{tab:main_results}), fix-g
achieves perfect accuracy on AN and LS and the highest accuracy
among all evaluated methods on the real-world \textit{Tue} dataset.
Across all 24 benchmarks, fix-g remains competitive with the
strongest baselines under the common default-parameter protocol.
The additional flexibility of est-g is particularly beneficial
under strong nonlinear observation distortions.

\subsection{Additional experimental analyses}

\textbf{Sensitivity to optimization and estimation settings.} Additional initializations yield only modest accuracy gains, and reconstruction scores and directional decisions stabilize after a few outer iterations in a representative example (Appendix~\ref{appendix:optimization_sensitivity}). Analyses of sample size, grid resolution, and quantile range show stable performance near the default settings, although overly coarse discretization reduces accuracy (Appendix~\ref{appendix:discretization_sensitivity}). 

\textbf{Directional score separation.} Across the 42 custom benchmark settings, normalized gaps between forward and reverse LRQS scores are typically separated from zero. Incorrect decisions tend to have smaller absolute gaps than correct decisions, indicating weaker directional separation (Appendix~\ref{appendix:directional_score_gap}). 

\textbf{Role of the low-rank constraint.} Removing the effective rank constraint yields zero reconstruction scores in both directions for every pair in the two evaluated settings. The low-rank restriction is therefore essential for directional discrimination in these experiments (Appendix~\ref{appendix:low_rank_ablation}). 

\textbf{Non-monotone transformations.} Under a family of non-monotone transformations, CVEL achieves perfect accuracy, while both LRQS variants achieve perfect accuracy for $K_{\mathrm{fit}}\in\{1,2,3,4\}$. The remaining baselines achieve substantially lower accuracy. These results demonstrate that CVEL and LRQS are empirically robust to this particular violation of the increasing-transformation assumption (Appendix~\ref{appendix:monotonicity_violation}).

\section{Limitations and conclusions}
\label{sec:limitations-conclusion}

\textbf{Limitations.}
LRQS is a bivariate method and does not address multivariate graphs, latent
confounding, or selection bias. Its generic-identifiability analysis characterizes exceptional cause
marginals through local restrictions. The theory covers any prescribed finite reverse quantile basis, while the free-basis result is proved for one reverse component. The
practical score also depends on conditional quantile estimation, the chosen number of non-intercept components
$K_{\mathrm{fit}}$, discretization, and a nonconvex alternating fit. 

\textbf{Conclusion.}
We introduced Low-Rank Quantile Surfaces, a model in which
the causal-direction conditional quantile surface becomes low rank after an
unknown monotone transformation. LRQS extends location-scale and post-nonlinear
heteroscedastic noise models, possesses generic identifiability,
and leads to a simple nonparametric causal score. Experiments show strong
performance under higher-rank shape variation and nonlinear observation
distortions, suggesting latent low-rank quantile structure as an effective
inductive bias for causal discovery.

\section*{Acknowledgment}
This work was partially supported by the Japan Science and Technology Agency (JST) under CREST Grant Number JPMJCR22D2 and by the Japan Society for the Promotion of Science (JSPS) under KAKENHI Grant Number JP24K20741. The authors declare no competing interests.

\newpage
\appendix

\section{Detailed related work}
\label{appendix:related_works}

\paragraph{Bivariate cause-effect identification.}
Because $X\to Y$ and $Y\to X$ impose no distinct
conditional-independence constraints, bivariate observational discovery
requires additional asymmetry assumptions
\citep{mooij2016distinguishing,guyon2019cause}.

\paragraph{Transport, velocity, and score-based approaches.}
DIVOT connects functional causal models to optimal transport and derives
directional criteria from transport dynamics \citep{tu2022optimal}.
Causal velocity models (CVEL) treat the cause as a time-like parameter,
relating counterfactual velocity fields to the joint score function
\citep{xi2025distinguishing}.
For multivariate nonlinear additive Gaussian-noise models, SCORE recovers
a causal ordering from the Hessian of the joint log-density
\citep{rolland22a}.

\paragraph{Generative, supervised, and likelihood-based methods.}
Causal generative neural networks (CGNNs) compare generative fits using
maximum mean discrepancy \citep{goudet2018learning}, whereas the
randomized causation coefficient learns causal decisions from labeled
cause-effect pairs \citep{lopez2015randomized}.
LOCI provides feature-map and neural-network estimators for LSNMs
\citep{Immer_2023}, and QPE-f estimates quantile partial effects using
normalizing flows \citep{QPE_2026causal}.
Bayesian model selection offers a complementary approach:
\citet{dhir2024bivariate} compare flexible bivariate models by marginal
likelihood under priors expressing independent causal mechanisms.

\paragraph{Connections to econometric models.}
Single-index quantile regression assumes
$Q_{Y\mid\mathbf W=\mathbf w}(u)
=\varphi_u(\boldsymbol{\omega}_u^\top\mathbf w)$
for a covariate vector $\mathbf w$, restricting dependence to one
projection at each quantile \citep{kong2012single}.
LRQS requires
$h(Q(x,u))=a(x)+\sum_{k=1}^{K}b_k(x)q_k(u)$
for a common increasing transformation $h$, so the transformed curves
$u\mapsto h(Q(x,u))$ lie in the shared space
$\operatorname{span}\{1,q_1,\ldots,q_K\}$.
This restriction remains substantive for scalar $x$, where the
single-index form alone is unrestricted.
For $K=1$, LRQS has the form of a transformed location--scale model
\citep{kloodt2021identification}.
A version with vector covariates and linear location and scale appears
in \citet[Eq.~(5)]{kong2012single}, with positive scale ensuring the
quantile representation.
More generally, $h(Y)$ is linear in the unknown features
$a(X),b_1(X),\ldots,b_K(X)$, with coefficients
$1,q_1(U),\ldots,q_K(U)$ driven by a common
$U\sim\mathrm{Unif}(0,1)$ independent of $X$.
This parallels the random-coefficient interpretation of quantile
regression \citep{koenker2005quantile} and its panel
extensions \citep{graham2018quantile}.
Our focus is when the transformed finite-rank restriction
distinguishes the two causal directions.

\paragraph{Broader connections.}
Additive models address unobserved causal and backdoor
paths \citep{pham_2026_aistats} and graph integration across
non-identical variable sets \citep{suzuki_2026};
location-scale models additionally accommodate heteroscedasticity
with hidden variables \citep{khan_2026_arxiv}.
These approaches build on regression and residual-independence
criteria.
Related downstream tasks include root-cause analysis of time-series
prediction errors \citep{yokoyama_2025} and optimal-transport-based
counterfactual distribution estimation \citep{pham_2024}, which targets
outcome distributions rather than causal direction.

\section{Beyond latent location-scale mechanisms}\label{appendix:beyond_pnl_hnm}
For $K = 1$, $h(Q(x,u)) = a(x) + b(x)q(u)$: after removing location and scale, every latent conditional quantile curve $u \mapsto h(Q(x,u))$ has the same standardized shape $q(u)$. Hence conditional skewness, tail ratios, and other location-scale invariants (when the corresponding moments exist) cannot vary with the cause $x$. (A nonlinear $g$ can still make the \emph{observed} quantile curve $u \mapsto Q(x,u)$ change shape with $x$, but only in the restricted way expressible by a PNL-HNM, a location-scale family pushed through one fixed nonlinearity.)

For $K = 2$, $h(Q(x,u)) = a(x) + b_1(x)q_1(u) + b_2(x)q_2(u)$: the relative coefficient $\lambda(x) = b_2(x)/b_1(x)$ can vary with $x$, so the standardized latent conditional quantile curve $u \mapsto h(Q(x,u))$ can change by shape variations beyond location and scale on this latent scale. For example:
\begin{itemize}
\item if $q_1$ is a symmetric Gaussian quantile and $q_2$ is an asymmetric Gumbel quantile, varying $\lambda(x)$ changes conditional skewness;
\item if $q_1$ 
 is Gaussian and $q_2$ a heavy-tailed Student-$t$ quantile, varying $\lambda(x)$ changes relative tail thickness.
\end{itemize}

\section{Theoretical details}\label{appendix:proof}
\subsection{Proof of Proposition~\ref{proposition:factorization}}
For fixed $x$, define $Z_x(u)=a(x)+\sum_{k=1}^{K}b_k(x)q_k(u)$. By Assumption~\ref{assumption:monotone}, $Z_x$ is strictly increasing. Since $U\sim \text{Unif}(0,1)$, the conditional quantile of $Z_x(U)$ is $Z_x(u)$. Because $g$ is strictly increasing, monotone transformations preserve quantile order, so $Q_{Y\mid X=x}(u) = g(Z_x(u))$. 

\subsection{Proof of Theorem~\ref{theorem:general_K_fixed_quantile}}\label{appendix:proof_theorem_1}
We first establish two consequences of Bayes’ rule used in both theorem proofs. The first relates the reverse drift ratio to the forward conditional score; the second reconstructs the cause log-density from a reverse quantile curve.

\begin{lemma}
\label{lem:bayes}
Write $c(y) = \partial_y \log  p_{Y}(y)$, and $\alpha(x,y) = \partial_y \log r (y \mid x)$. Then
\begin{equation}
\label{eq:RuBayes}
R_u(y,u)=c(y)-\alpha(Q^\leftarrow(y,u),y),
\end{equation}
\begin{equation}
\label{eq:xi-recon}
\xi(Q^{\leftarrow}(y,u))
=
\log  p_{Y}(y)-\log r(y\mid Q^{\leftarrow}(y,u))-\log Q_u^{\leftarrow}(y,u).
\end{equation}
\end{lemma}

\begin{proof}
Let $F(x,y)=F_{X\mid Y=y}(x)$ and $p^{\leftarrow}(x\mid y)=p_{X\mid Y=y}(x)$.
Differentiating $F(Q^\leftarrow(y,u),y)=u$ gives
$Q^\leftarrow_u=1/p^\leftarrow(Q^\leftarrow \mid y)$ and
$R=-F_y(Q^\leftarrow,y)$, where subscripts on $F$ and $p^\leftarrow$
denote partial derivatives before evaluation at $x=Q^\leftarrow(y,u)$.
Consequently, $R_u=-p^{\leftarrow}_y(Q^\leftarrow \mid y)/p^{\leftarrow}(Q^\leftarrow \mid y)$.
Substituting Bayes' formula
$p^{\leftarrow}(x \mid y)=e^{\xi(x)}r(y\mid x)/ p_{Y}(y)$ yields Eq.~(\ref{eq:RuBayes}); taking logarithms in the identity
for $Q^\leftarrow_u$ gives Eq.~(\ref{eq:xi-recon}).
\end{proof}

Fix an interior point $(x_0,y_0)$ satisfying
Assumptions~\ref{assumption:anchor} and~\ref{assumption:smooth},
and set $x(u)=Q^\leftarrow(y_0,u)$.
Choose an interval $J_x$ around $x_0$ on which
$\alpha_{y_0}'$ remains nonzero, and a quantile interval
$I_u$ around $u_0$ such that $x(I_u)\subset J_x$.
Set $I_x=x(I_u)$.
Below, $\alpha_{y_0}^{-1}$ denotes the inverse of
$\alpha_{y_0}$ restricted to $J_x$.
At $y_0$, the reverse LRQS formula for $R$ in Eq.~(\ref{eq:backward_drift_ratio}) involves
$2L+1$ coefficients. Its value is unchanged by a common nonzero
scaling of these coefficients. Since at least one denominator
coefficient is nonzero, normalizing that coefficient leaves at most
$2L$ real parameters. Denote them collectively by $\theta$ and the
resulting function by $R_\theta(u)$.

Set $c_0=c(y_0)$.
Equation~(\ref{eq:RuBayes}) and the local invertibility
of $\alpha_{y_0}$ give
\begin{equation}
x(u)=\alpha_{y_0}^{-1}\bigl(c_0-R_\theta'(u)\bigr),
\qquad u\in I_u.
\label{eq:lrqs-app-fixed-inverse}
\end{equation}
Thus both $x(u)$ and its derivative $x_u(u)$ are determined by
$(\theta,c_0)$ for fixed $y_0$.

With $\ell_0=\log p_Y(y_0)$,
Eq.~(\ref{eq:xi-recon}) becomes
$\xi(x(u))=\ell_0-\log r(y_0\mid x(u))-\log x_u(u)$.
Since $x_u>0$, the map $u\mapsto x(u)$ is invertible locally, so this
identity determines $\xi|_{I_x}$ without introducing another
unknown function. The local inverse branches are determined by
the fixed mechanism $r$. Counting the at most $2L$ coefficients in
$\theta$, the two scalars $c_0,\ell_0$, and the anchor $y_0$ gives
the claimed upper bound of $2L+3$ real parameters.

\subsection{Proof of Theorem~\ref{theorem:K__1_free_quantile}}
\label{appendix:proof_theorem_2}

\begin{proof}
Fix $\xi\in\mathcal B(r;\mathcal M_{1,\text{free}})$. Choose an interior point satisfying
Assumptions~\ref{assumption:anchor}
and~\ref{assumption:smooth}, and let $I_u$ be a sufficiently small
quantile interval around the corresponding $u_0$.
By Eq.~(\ref{eq:RuBayes}),
$R_{uu}(y_0,u)=
-\alpha_{y_0}'(Q^\leftarrow(y_0,u))Q_u^\leftarrow(y_0,u)$
is nonzero on $I_u$.
Hence $R(y_0,\cdot)$ is not identically zero there.
After choosing a smaller interval, which need not contain $u_0$, we may therefore assume
$R(y_0,u)\ne0$ throughout.
Set $I_x=Q^\leftarrow(y_0,I_u)$.

Use the functions $x(u)=Q^{\leftarrow}(y_0,u)$,
$R(u)=R(y_0,u)$, and $D(u)=R_y(y_0,u)$. Set $c_0=(\log p_Y)'(y_0)$,
$c_1=(\log p_Y)''(y_0)$, $\rho_0=\rho(y_0)$, and
$\rho_1=\rho'(y_0)$.
Also define
$\beta_{y_0}(x)=\left.\partial_y^2\log r(y\mid x)\right|_{y=y_0}$.
Both $\alpha_{y_0}$ and $\beta_{y_0}$ are fixed by $r$ and $y_0$;
primes on these functions denote differentiation in $x$.

To obtain an equation for $D_u$, we differentiate
Eq.~(\ref{eq:RuBayes}) with respect to $y$ while holding $u$ fixed.
Since $y$ enters both arguments of $\alpha(Q^\leftarrow(y,u),y)$,
the chain rule gives
\[
\partial_y\bigl[\alpha(Q^\leftarrow(y,u),y)\bigr]
=
\alpha_x(Q^\leftarrow(y,u),y)\,Q_y^\leftarrow(y,u)
+\alpha_y(Q^\leftarrow(y,u),y).
\]
At $y=y_0$, the quantile derivative is
$Q_y^\leftarrow(y_0,u)=R(u)x_u(u)$, while
$\alpha_x(x,y_0)=\alpha_{y_0}'(x)$ and
$\alpha_y(x,y_0)=\beta_{y_0}(x)$.
Thus, evaluating the Bayes identity and its $y$-derivative at
$y=y_0$, and using $D_u(u)=R_{yu}(y_0,u)$, gives
\begin{align}
R_u&=c_0-\alpha_{y_0}(x),
\label{eq:R_u_anchor}\\
D_u&=c_1-\beta_{y_0}(x)-\alpha_{y_0}'(x)R x_u.
\label{eq:R_yu_anchor}
\end{align}

On the other hand, Eq.~(\ref{eq:ODE_R}) and its $y$-derivative give
\[
R_u+TR=\rho_0,\qquad D_u+TD=\rho_1,
\]
because $T$ depends only on $u$.
Together with Eq.~(\ref{eq:R_u_anchor}), the first identity yields
$T=(\rho_0-c_0+\alpha_{y_0}(x))/R$.
Substituting this into the second identity and comparing with
Eq.~(\ref{eq:R_yu_anchor}) gives
\begin{equation}
\label{eq:main-system}
\begin{aligned}
x_u&=
\frac{R\bigl(c_1-\beta_{y_0}(x)-\rho_1\bigr)
      +\bigl(\rho_0-c_0+\alpha_{y_0}(x)\bigr)D}
     {\alpha_{y_0}'(x)R^2},\\
R_u&=c_0-\alpha_{y_0}(x),\\
D_u&=\rho_1-
\frac{\bigl(\rho_0-c_0+\alpha_{y_0}(x)\bigr)D}{R}.
\end{aligned}
\end{equation}
Thus no unknown quantile-basis function remains in the system.
Its right-hand side is $C^1$ wherever $R\ne0$ and
$\alpha_{y_0}'(x)\ne0$, as ensured by
Assumptions~\ref{assumption:anchor} and~\ref{assumption:smooth}.
For fixed $y_0$ and a chosen $u_\ast\in I_u$, local uniqueness
therefore determines $(x,R,D)$ from the three initial values
$(x_\ast,R_\ast,D_\ast)=(x,R,D)(u_\ast)$ and the four scalars
$(c_0,c_1,\rho_0,\rho_1)$.

On the actual solution, $x_u=Q_u^{\leftarrow}(y_0,u)>0$.
With $\ell_0=\log p_Y(y_0)$, Eq.~(\ref{eq:xi-recon}) becomes
$\xi(x(u))=\ell_0-\log r(y_0\mid x(u))-\log x_u(u)$.
Since $u\mapsto x(u)$ is locally invertible, this determines
$\xi|_{I_x}$, after shrinking $I_u$ and $I_x=x(I_u)$ if necessary.
Neither the system nor this reconstruction depends explicitly on
$u$. Hence changing the origin of the local $u$-parameter only
reparametrizes the same function of $x$, and $u_\ast$ contributes
no additional parameter.

Counting $(x_\ast,R_\ast,D_\ast,c_0,c_1,\rho_0,\rho_1,\ell_0,y_0)$
gives at most nine real parameters.
Allowing their admissible values defines a family depending only
on $r$ and containing all the required restrictions.
Additional conditions for a valid reverse-compatible joint density
can only restrict this family.
\end{proof}

\subsection{Scope and exclusions of the local assumptions}
\label{appendix:assumption_scope}

We explain the roles of
Assumptions~\ref{assumption:anchor} and~\ref{assumption:smooth}
and derive several cases outside their joint scope.
All density derivatives below are evaluated in open neighborhoods
where the relevant densities are positive and the required
derivatives exist.

\paragraph{Implications of Assumption~\ref{assumption:smooth}.}
Purely discrete, deterministic, or singular conditional laws are
outside this setting.
The assumption is local: bounded support or nonsmoothness away
from the selected neighborhoods does not by itself violate it.
Nor is boundedness of quantile derivatives as $u\to0$ or $u\to1$
required.

\paragraph{Assumption~\ref{assumption:anchor} is symmetric.}
Write $\alpha(x,y)=\partial_y\log r(y\mid x)$, so that
$\alpha_x(x_0,y_0)=\alpha_{y_0}'(x_0)$.
The two factorizations of the joint density give
\begin{equation}
\label{eq:scope-symmetry}
\alpha_x(x,y)
=
\partial_x\partial_y\log p_{X,Y}(x,y)
=
\partial_y\partial_x\log p_{X\mid Y=y}(x).
\end{equation}
Thus Assumption~\ref{assumption:anchor} is symmetric in $X$ and $Y$,
despite being stated through the forward conditional density.
Moreover, differentiating Eq.~(\ref{eq:RuBayes}) in $u$ gives
\begin{equation}
\label{eq:scope-drift}
R_{uu}(y,u)
=
-\alpha_x(Q^\leftarrow(y,u),y)\,
 Q_u^\leftarrow(y,u).
\end{equation}
Because $Q_u^\leftarrow>0$ at regular points, the nondegeneracy
condition $\alpha_x(x_0,y_0) \ne 0$ is equivalent to $R_{uu}(y_0,u_0) \ne 0$, i.e., nonzero curvature of the reverse drift ratio in the quantile coordinate at the corresponding point.

The condition detects dependence within smooth regions of the
density.
Indeed, on an open rectangle, $\alpha_x\equiv0$ implies
$\log p_{X,Y}(x,y)=A(x)+B(y)$ by integration.
If the density is positive and $C^2$ throughout the interior of
a rectangular support, this factorization implies independence.
Consequently, every dependent distribution in that setting has
a point where Assumption~\ref{assumption:anchor} holds.

\paragraph{For PNL-HNMs.}
Consider $Y=g(a(X)+b(X)\varepsilon)$, with $\varepsilon\cind X$,
$b>0$, and $h=g^{-1}$.
Work locally where the functions are sufficiently differentiable
and $h'>0$.
Set $e=(h(y)-a(x))/b(x)$ and $\nu=\log p_\varepsilon$.
The conditional density satisfies
$r(y\mid x)=p_\varepsilon(e)h'(y)/b(x)$.
Differentiating its logarithm with respect to $y$, using
$e_y=h'(y)/b(x)$, gives
\begin{equation}
\label{eq:alpha}
\alpha(x,y)
=
\frac{h'(y)}{b(x)}\nu'(e)
+\frac{h''(y)}{h'(y)}.
\end{equation}
Differentiating with respect to $x$, using
$e_x=-(a'(x)+b'(x)e)/b(x)$, then gives
\begin{equation}
\label{eq:scope-locationscale}
\alpha_x(x,y)
=
-\frac{h'(y)}{b(x)^2}
\left[
\bigl(a'(x)+b'(x)e\bigr)\nu''(e)
+b'(x)\nu'(e)
\right].
\end{equation}
The outer transformation contributes only the positive factor
$h'(y)$, so it cannot remove the degeneracies identified below.
By Eq.~(\ref{eq:scope-symmetry}), the same conclusions apply to
reverse representations after exchanging the variables.
\paragraph{For Gaussian noise.} Standard Gaussian noise gives
$\alpha_x(x,y)=h'(y)(a'(x)+2b'(x)e)/b(x)^2$.
At a conditioning value where $a'$ and $b'$ are not both zero,
this expression is nonzero for some noise value.

\paragraph{Uniform noise, including varying scale.}
For uniform noise on any nondegenerate interval,
$\nu'=\nu''=0$ throughout the support interior.
Equation~(\ref{eq:scope-locationscale}) therefore gives
$\alpha_x=0$, regardless of the location function, scale function,
or monotone transformation.
Hence no point satisfies both assumptions.

In the reverse direction, for uniform reverse noise, $\tilde q$ is affine, so
$T=\tilde q''/\tilde q'=0$. Equation~(\ref{eq:ODE_R}) therefore gives
$R_u=\rho(y)$, hence $R_{uu}=0$. Equation~(\ref{eq:scope-drift}) again gives $\alpha_x=0$.
Thus a forward mechanism satisfying
Assumption~\ref{assumption:anchor} cannot admit a regular
uniform-noise reverse PNL-HNM.
\paragraph{Constant-scale exponential and Laplace noise.}
For standard exponential noise, $\nu'(e)=-1$ and $\nu''(e)=0$
on $e>0$.
Thus
$\alpha_x(x,y)=h'(y)b'(x)/b(x)^2$.
When $b$ is constant, Assumption~\ref{assumption:anchor} fails
throughout the support interior, even if $a$ is nonconstant.

For standard Laplace noise,
$\nu(e)=-|e|-\log 2$, so $\nu''(e)=0$ away from $e=0$.
With constant $b$, Eq.~(\ref{eq:scope-locationscale}) again gives
$\alpha_x=0$ at all smooth points.
At $e=0$, the density has a kink and the required local regularity
fails.
Hence constant-scale Laplace models have no point satisfying
both assumptions.

Unlike uniform noise, these exclusions need not persist under varying scale.
For Laplace noise, away from its kink,
$\alpha_x(x,y)=
h'(y)b'(x)\operatorname{sign}(e)/b(x)^2$.
Thus exponential or Laplace noise can satisfy
Assumption~\ref{assumption:anchor} at a regular point where
$b'(x)\ne0$.

\paragraph{Pure-scale power-law noise models.}
For a pure-scale mechanism with $a\equiv \mathrm{const}$, suppose the noise
log-density has the form
$\nu(t)=\mathrm{const}+\gamma\log t$
on its positive support.
Then $t\nu''(t)+\nu'(t)=0$, and
Eq.~(\ref{eq:scope-locationscale}) gives $\alpha_x=0$
for every scale function $b$.

Examples include
$\varepsilon\sim\mathrm{Beta}(\eta,1)$, with density
$p_\varepsilon(t)=\eta t^{\eta-1}$ on $0<t<1$, and Pareto noise
with density $p_\varepsilon(t)=\eta t^{-\eta-1}$ on $t>1$,
where $\eta>0$.
For nonconstant $b$, these mechanisms can describe dependent
variables, but their interior log-densities remain separable
in $x$ and $y$.
They therefore fall outside Assumption~\ref{assumption:anchor},
including after a monotone transformation.

\paragraph{Piecewise-affine reverse quantile bases.}
The exclusion also extends beyond one-component models.
Suppose all prescribed reverse bases are piecewise affine with
finitely many knots.
On any interval avoiding their combined knots, the reverse
latent surface has the form
$z^\leftarrow(y,u)=A(y)+B(y)u$.
At regular points $B(y)>0$, and
$R(y,u)=(A'(y)+B'(y)u)/B(y)$ is affine in $u$.
Hence Eq.~(\ref{eq:scope-drift}) gives $\alpha_x=0$.
Genuine knots violate the required smoothness; if a knot disappears in the resulting surface, continuity still gives
$R_{uu}=0$ there.
Thus exact reverse representations built from such
bases cannot satisfy the two assumptions jointly.

\paragraph{Connection to classical ANM results.}
For an ANM $Y=a(X)+\varepsilon$, with $\varepsilon\cind X$ and
$\nu=\log p_\varepsilon$, direct differentiation gives
$\alpha_x(x,y)=-a'(x)\nu''(y-a(x))$.
The differential-equation arguments of
\citet[Theorem~1]{Hoyer09NIPS} and
\citet[Condition~19 and Proposition~21]{Peters14JMLR}
use this nonvanishing product; their finite-dimensional genericity
statements additionally require it to be nonzero for all but
countably many cause values along some fixed conditioning slice.
These particular genericity results do not cover uniform,
exponential, or Laplace noise: $\nu''$ vanishes on every smooth
piece of the positive-density region, while the required
differentiability fails at the Laplace kink.
Our condition requires only one regular point satisfying both
assumptions, but shares this nondegeneracy ingredient.

\section{Additional experimental results}

\subsection{Additional results}
\label{appendix:experiment_full_result}

We provide the complete experimental results across all 42 custom mechanism
variants and the full suite of 24 existing bivariate benchmark datasets.
For the existing benchmarks, we report both a controlled comparison using a
single default configuration for each method and a complementary comparison
under dataset-specific hyperparameter tuning.

\textbf{Comprehensive evaluation on structural complexity.}
Table~\ref{tab:appendix_multirank} details the performance under increasing
structural complexity ($K_{\mathrm{gen}} \in \{2,3,4,5\}$) across Gaussian, uniform, and beta
noise distributions. LRQS remains consistently strong across the different
complexity levels and noise distributions. CVEL is also highly competitive in
several higher-complexity settings, whereas QPE-k and the QPE-f variants show
larger performance degradation in a number of multirank settings. Overall,
the results support the robustness of the low-rank quantile surface
representation to substantial conditional shape variation.

\textbf{Consistent robustness against nonlinear distortions.}
Tables~\ref{tab:appendix_pnl_an} and \ref{tab:appendix_pnl_hnm} present the
complete results under nonlinear observation distortions for both post-nonlinear additive
noise models (PNL-AN) and post-nonlinear heteroscedastic noise models (PNL-HNM).
The \textit{est-g} variant remains highly accurate across a broad range of
noise distributions and nonlinear transformations, with particularly strong
performance under severe nonlinear distortions. While several competing
methods are effective in specific settings, the overall results demonstrate
the benefit of explicitly estimating the unknown monotone transformation when
observation-level distortions are substantial.

\textbf{Extended results on existing benchmarks.}
Tables~\ref{tab:appendix_qpe_benchmarks_1} and
\ref{tab:appendix_qpe_benchmarks_2} report the complete results on all 24
bivariate benchmark datasets using the same default-parameter protocol as in
Table~\ref{tab:main_results}. A single default configuration for each method is
used across all datasets, without dataset-specific tuning. Under this
controlled setting, LRQS (\textit{fix-g}) remains competitive with the
strongest baselines while retaining its low computational cost.

For completeness, Tables~\ref{tab:appendix_qpe_benchmarks_tuned_1} and
\ref{tab:appendix_qpe_benchmarks_tuned_2} additionally report a
dataset-specific tuning comparison following the evaluation style of
\citet{QPE_2026causal}. Baseline values are cited from
\citet{QPE_2026causal}, while LRQS is evaluated using dataset-specific
hyperparameter tuning. This complementary comparison shows the performance
achievable when hyperparameters are adapted to individual datasets, whereas
the default-parameter comparison above provides a more controlled assessment
under a common evaluation protocol.

\begin{table*}[!ht]
\centering
\caption{
Detailed accuracy comparison on multirank structural-complexity benchmarks.
Gaussian, uniform, and beta denote the noise distributions.
}
\label{tab:appendix_multirank}
\resizebox{\textwidth}{!}{
\begin{tabular}{lccccccccccccc}
\toprule
Method
& \multicolumn{4}{c}{Gaussian}
& \multicolumn{4}{c}{Uniform}
& \multicolumn{4}{c}{Beta}
& Time (s) \\
\cmidrule(lr){2-5}
\cmidrule(lr){6-9}
\cmidrule(lr){10-13}
& $K_{\mathrm{gen}}=2$ & $K_{\mathrm{gen}}=3$ & $K_{\mathrm{gen}}=4$ & $K_{\mathrm{gen}}=5$
& $K_{\mathrm{gen}}=2$ & $K_{\mathrm{gen}}=3$ & $K_{\mathrm{gen}}=4$ & $K_{\mathrm{gen}}=5$
& $K_{\mathrm{gen}}=2$ & $K_{\mathrm{gen}}=3$ & $K_{\mathrm{gen}}=4$ & $K_{\mathrm{gen}}=5$
& \\
\midrule
ANM
& 0.44 & 0.48 & 0.43 & 0.40
& 0.39 & 0.49 & 0.59 & 0.60
& 0.46 & 0.58 & 0.50 & 0.54
& 8.420 \\

DIVOT
& 0.04 & 0.00 & 0.00 & 0.00
& 0.18 & 0.03 & 0.00 & 0.00
& 0.04 & 0.00 & 0.00 & 0.00
& 1.348 \\

CVEL
& 0.59 & 0.91 & 0.98 & \textbf{1.00}
& 0.43 & 0.80 & \textbf{0.98} & \textbf{1.00}
& 0.62 & \textbf{0.94} & \textbf{0.98} & \textbf{1.00}
& 1.994 \\

QPE-k
& 0.36 & 0.13 & 0.09 & 0.11
& \textbf{0.98} & 0.77 & 0.47 & 0.23
& 0.69 & 0.22 & 0.07 & 0.04
& 0.057 \\

QPE-f (fixed)
& 0.34 & 0.28 & 0.40 & 0.37
& 0.53 & 0.30 & 0.30 & 0.29
& 0.46 & 0.25 & 0.27 & 0.20
& 37.048 \\

QPE-f (poly)
& 0.33 & 0.18 & 0.17 & 0.15
& 0.54 & 0.35 & 0.31 & 0.37
& 0.49 & 0.29 & 0.30 & 0.21
& 27.550 \\

QPE-f (lowrank)
& 0.36 & 0.37 & 0.63 & 0.64
& 0.45 & 0.30 & 0.53 & 0.64
& 0.38 & 0.39 & 0.57 & 0.58
& 38.597 \\

LRQS (est-g)
& 0.94 & \textbf{0.97} & 0.96 & \textbf{1.00}
& 0.96 & \textbf{0.93} & 0.82 & 0.71
& \textbf{0.95} & 0.88 & 0.90 & 0.91
& 2.038 \\

LRQS (fix-g)
& \textbf{0.95} & \textbf{0.97} & \textbf{0.99} & \textbf{1.00}
& \textbf{0.98} & 0.84 & 0.69 & 0.73
& 0.94 & 0.88 & 0.89 & 0.98
& \textbf{0.046} \\
\bottomrule
\end{tabular}
}
\end{table*}

\begin{table*}[!ht]
\centering
\caption{
Detailed accuracy comparison on PNL-AN benchmarks.
Gaussian, uniform, and beta denote the noise distributions.
Id, Sig, Exp, and Tanh denote identity, sigmoid, exponential, and hyperbolic tangent transformations.
}
\label{tab:appendix_pnl_an}
\resizebox{\textwidth}{!}{
\begin{tabular}{lcccccccccccccccc}
\toprule
Method
& \multicolumn{5}{c}{Gaussian}
& \multicolumn{5}{c}{Uniform}
& \multicolumn{5}{c}{Beta}
& Time (s) \\
\cmidrule(lr){2-6}
\cmidrule(lr){7-11}
\cmidrule(lr){12-16}
& Id & Cube & Sig & Exp & Tanh
& Id & Cube & Sig & Exp & Tanh
& Id & Cube & Sig & Exp & Tanh
& \\
\midrule
ANM
& \textbf{1.00} & 0.23 & 0.63 & 0.21 & 0.27
& \textbf{1.00} & 0.20 & 0.58 & 0.20 & 0.17
& \textbf{1.00} & 0.20 & 0.59 & 0.17 & 0.24
& 5.795 \\

DIVOT
& 0.96 & 0.00 & 0.89 & 0.00 & 0.31
& \textbf{1.00} & 0.00 & 0.91 & 0.00 & 0.24
& \textbf{1.00} & 0.00 & 0.93 & 0.00 & 0.33
& 1.127 \\

CVEL
& 0.02 & \textbf{1.00} & 0.04 & \textbf{1.00} & 0.70
& 0.76 & \textbf{0.99} & 0.87 & \textbf{1.00} & \textbf{1.00}
& 0.19 & \textbf{1.00} & 0.52 & \textbf{1.00} & \textbf{0.99}
& 2.178 \\

QPE-k
& 0.90 & 0.34 & 0.75 & 0.16 & 0.10
& \textbf{1.00} & 0.74 & \textbf{1.00} & 0.49 & 0.43
& \textbf{1.00} & 0.60 & \textbf{0.97} & 0.29 & 0.25
& 0.067 \\

QPE-f (fixed)
& 0.83 & 0.35 & 0.24 & 0.68 & 0.13
& 0.27 & 0.13 & 0.30 & 0.55 & 0.23
& 0.33 & 0.26 & 0.17 & 0.70 & 0.24
& 32.326 \\

QPE-f (poly)
& 0.87 & 0.14 & 0.15 & 0.27 & 0.07
& 0.37 & 0.12 & 0.34 & 0.49 & 0.30
& 0.41 & 0.14 & 0.17 & 0.49 & 0.28
& 34.431 \\

QPE-f (lowrank)
& 0.77 & 0.43 & 0.48 & 0.61 & 0.31
& 0.60 & 0.22 & 0.64 & 0.72 & 0.51
& 0.51 & 0.29 & 0.50 & 0.69 & 0.52
& 37.496 \\

LRQS (est-g)
& 0.90 & 0.99 & 0.95 & \textbf{1.00} & \textbf{0.91}
& 0.93 & 0.95 & 0.94 & 0.98 & 0.93
& 0.96 & 0.96 & 0.95 & 0.98 & 0.89
& 2.356 \\

LRQS (fix-g)
& 0.94 & 0.82 & \textbf{0.96} & \textbf{1.00} & 0.59
& 0.98 & 0.56 & 0.91 & 0.97 & 0.69
& 0.93 & 0.58 & 0.95 & \textbf{1.00} & 0.74
& \textbf{0.051} \\
\bottomrule
\end{tabular}
}
\end{table*}

\begin{table*}[!ht]
\centering
\caption{
Detailed accuracy comparison on PNL-HNM benchmarks.
Gaussian, uniform, and beta denote the noise distributions.
Id, Sig, Exp, and Tanh denote identity, sigmoid, exponential, and hyperbolic tangent transformations.
}
\label{tab:appendix_pnl_hnm}
\resizebox{\textwidth}{!}{
\begin{tabular}{lcccccccccccccccc}
\toprule
Method
& \multicolumn{5}{c}{Gaussian}
& \multicolumn{5}{c}{Uniform}
& \multicolumn{5}{c}{Beta}
& Time (s) \\
\cmidrule(lr){2-6}
\cmidrule(lr){7-11}
\cmidrule(lr){12-16}
& Id & Cube & Sig & Exp & Tanh
& Id & Cube & Sig & Exp & Tanh
& Id & Cube & Sig & Exp & Tanh
& \\
\midrule
ANM
& 0.36 & 0.34 & 0.35 & 0.26 & 0.33
& 0.35 & 0.36 & 0.34 & 0.32 & 0.36
& 0.41 & 0.31 & 0.36 & 0.32 & 0.34
& 8.006 \\

DIVOT
& 0.25 & 0.00 & 0.36 & 0.00 & 0.08
& 0.30 & 0.00 & 0.32 & 0.02 & 0.10
& 0.36 & 0.00 & 0.34 & 0.04 & 0.09
& 1.324 \\

CVEL
& 0.18 & 0.99 & 0.20 & 0.94 & 0.66
& 0.24 & 0.99 & 0.39 & 0.95 & 0.74
& 0.17 & \textbf{0.98} & 0.22 & 0.92 & 0.76
& 1.982 \\

QPE-k
& 0.94 & 0.48 & 0.79 & 0.32 & 0.45
& \textbf{1.00} & 0.76 & 0.94 & 0.65 & 0.60
& 0.98 & 0.63 & 0.91 & 0.39 & 0.53
& 0.060 \\

QPE-f (fixed)
& 0.92 & 0.68 & 0.72 & 0.90 & 0.47
& 0.45 & 0.73 & 0.31 & 0.79 & 0.34
& 0.60 & 0.69 & 0.46 & 0.79 & 0.34
& 34.736 \\

QPE-f (poly)
& 0.90 & 0.55 & 0.65 & 0.53 & 0.44
& 0.50 & 0.58 & 0.30 & 0.66 & 0.40
& 0.62 & 0.58 & 0.41 & 0.62 & 0.36
& 34.774 \\

QPE-f (lowrank)
& 0.59 & 0.73 & 0.56 & 0.85 & 0.54
& 0.36 & 0.75 & 0.38 & 0.76 & 0.46
& 0.41 & 0.82 & 0.37 & 0.75 & 0.52
& 37.930 \\

LRQS (est-g)
& \textbf{1.00} & \textbf{1.00} & \textbf{0.99} & \textbf{1.00} & \textbf{0.96}
& \textbf{1.00} & \textbf{1.00} & \textbf{1.00} & 0.98 & \textbf{0.99}
& \textbf{0.99} & \textbf{0.98} & \textbf{0.99} & \textbf{1.00} & \textbf{0.97}
& 2.234 \\

LRQS (fix-g)
& \textbf{1.00} & 0.97 & 0.95 & \textbf{1.00} & 0.78
& \textbf{1.00} & 0.95 & 0.97 & \textbf{1.00} & 0.83
& \textbf{0.99} & 0.88 & 0.97 & \textbf{1.00} & 0.77
& \textbf{0.051} \\
\bottomrule
\end{tabular}
}
\end{table*}

\begin{table*}[!ht]
\centering
\caption{
Detailed accuracy comparison on the first group of bivariate benchmark datasets.
All methods are evaluated using a single default configuration across datasets.
}
\label{tab:appendix_qpe_benchmarks_1}
\resizebox{\textwidth}{!}{
\begin{tabular}{lcccccccccccc}
\toprule
Method
& AN & AN-s & LS & LS-s & MNU & SIM & SIM-c & SIM-g & SIM-ln & Tue & Cha & Net \\
\midrule

ANM
& \textbf{1.00} & \textbf{1.00} & 0.42 & 0.25 & 0.25
& 0.74 & \textbf{0.79} & 0.68 & 0.72 & 0.55
& \textbf{0.73} & 0.73 \\

DIVOT
& \textbf{1.00} & \textbf{1.00} & 0.72 & 0.34 & 0.91
& 0.73 & 0.70 & 0.68 & 0.60 & 0.46
& 0.52 & 0.80 \\

CVEL
& 0.25 & 0.09 & 0.18 & 0.25 & 0.70
& 0.64 & 0.62 & 0.80 & 0.65 & 0.34
& 0.67 & 0.51 \\

QPE-k
& 0.99 & 0.88 & \textbf{1.00} & 0.78 & \textbf{1.00}
& \textbf{0.83} & \textbf{0.79} & \textbf{0.83} & 0.68 & 0.54
& 0.60 & 0.89 \\

QPE-f (fixed)
& \textbf{1.00} & 0.94 & 0.98 & \textbf{1.00} & 0.93
& 0.74 & 0.70 & 0.54 & 0.79 & 0.61
& 0.54 & 0.89 \\

QPE-f (poly)
& 0.97 & \textbf{1.00} & 0.99 & \textbf{1.00} & 0.99
& 0.75 & 0.66 & 0.61 & \textbf{0.84} & 0.61
& 0.50 & \textbf{0.91} \\

QPE-f (lowrank)
& 0.51 & 0.17 & 0.52 & 0.56 & 0.89
& 0.74 & 0.66 & 0.55 & 0.72 & 0.48
& 0.55 & 0.63 \\

\midrule

LRQS (est-g)
& 0.98 & 0.81 & 0.96 & 0.73 & 0.74
& 0.61 & 0.52 & 0.71 & 0.56 & 0.69
& 0.61 & 0.71 \\

LRQS (fix-g)
& \textbf{1.00} & 0.93 & \textbf{1.00} & 0.92 & 0.84
& 0.69 & 0.77 & 0.72 & 0.64 & \textbf{0.78}
& 0.70 & 0.85 \\

\bottomrule
\end{tabular}
}
\end{table*}

\begin{table*}[!ht]
\centering
\caption{
Detailed accuracy comparison on the second group of bivariate benchmark datasets.
All methods are evaluated using a single default configuration across datasets.
}
\label{tab:appendix_qpe_benchmarks_2}
\resizebox{\textwidth}{!}{
\begin{tabular}{lcccccccccccc}
\toprule
Method
& Multi & D4-s1 & D4-s2a & D4-s2b & D4-s2c
& Per & Sig & Vex & Qd-V & Sig-V & RbF-V & NN-V \\
\midrule

ANM
& 0.60 & 0.58 & 0.62 & \textbf{0.61} & 0.58
& 0.63 & 0.28 & 0.17 & 0.82 & 0.76 & 0.47 & 0.64 \\

DIVOT
& 0.36 & \textbf{0.75} & 0.60 & 0.59 & 0.54
& 0.90 & 0.61 & 0.05 & 0.37 & 0.46 & 0.59 & 0.40 \\

CVEL
& 0.87 & 0.42 & 0.47 & 0.41 & 0.47
& \textbf{1.00} & 0.69 & 0.91 & \textbf{0.83}
& \textbf{0.90} & \textbf{0.90} & \textbf{0.88} \\

QPE-k
& \textbf{0.88} & 0.58 & \textbf{0.67} & \textbf{0.61}
& \textbf{0.64} & 0.77 & \textbf{0.89} & 0.63
& 0.42 & 0.67 & 0.68 & 0.53 \\

QPE-f (fixed)
& 0.76 & 0.54 & 0.54 & 0.53 & 0.47
& 0.95 & 0.47 & 0.93 & 0.76 & 0.69 & 0.68 & 0.75 \\

QPE-f (poly)
& 0.81 & 0.54 & 0.54 & 0.57 & 0.50
& 0.96 & 0.54 & 0.96 & 0.64 & 0.64 & 0.61 & 0.72 \\

QPE-f (lowrank)
& 0.70 & 0.33 & 0.66 & 0.53 & 0.55
& 0.83 & 0.73 & \textbf{0.97} & 0.73
& 0.67 & 0.69 & 0.77 \\

\midrule

LRQS (est-g)
& 0.84 & 0.50 & 0.53 & 0.51 & 0.51
& 0.64 & 0.54 & 0.78 & 0.57 & 0.60 & 0.44 & 0.61 \\

LRQS (fix-g)
& 0.86 & 0.42 & 0.64 & 0.54 & 0.53
& 0.74 & 0.52 & 0.70 & 0.73 & 0.79 & 0.65 & 0.79 \\

\bottomrule
\end{tabular}
}
\end{table*}

\begin{table*}[!ht]
\centering
\caption{
Detailed accuracy comparison on the first group of bivariate benchmark datasets
under dataset-specific tuning.
Baseline values are cited from \citet{QPE_2026causal}, while LRQS results are
obtained using dataset-specific hyperparameter tuning.
}
\label{tab:appendix_qpe_benchmarks_tuned_1}
\resizebox{\textwidth}{!}{
\begin{tabular}{lcccccccccccc}
\toprule
Method
& AN & AN-s & LS & LS-s & MNU & SIM & SIM-c & SIM-g & SIM-ln & Tue & Cha & Net \\
\midrule

ANM
& 0.43 & 0.47 & 0.46 & 0.45 & 0.40
& 0.45 & 0.49 & 0.41 & 0.46 & 0.65
& 0.41 & 0.47 \\

DIVOT
& 0.62 & 0.69 & 0.45 & 0.69 & \textbf{1.00}
& 0.68 & 0.47 & 0.60 & 0.63 & 0.38
& 0.44 & 0.49 \\

CVEL
& \textbf{1.00} & 0.98 & 0.98 & 0.93 & 0.94
& 0.63 & 0.72 & \textbf{0.90} & 0.76 & 0.64
& 0.68 & 0.62 \\

QPE-k
& 0.99 & 0.88 & \textbf{1.00} & 0.78 & \textbf{1.00}
& 0.83 & 0.79 & 0.83 & 0.68 & 0.54
& 0.60 & 0.89 \\

QPE-f
& \textbf{1.00} & \textbf{1.00} & \textbf{1.00} & \textbf{0.99}
& \textbf{1.00}
& \textbf{0.88} & \textbf{0.88} & 0.86 & \textbf{0.92} & 0.70
& \textbf{0.85} & 0.86 \\

\midrule

LRQS (est-g)
& \textbf{1.00} & \textbf{1.00} & \textbf{1.00} & 0.78 & 0.98
& 0.75 & 0.74 & 0.82 & 0.83 & 0.72
& 0.66 & 0.84 \\

LRQS (fix-g)
& \textbf{1.00} & \textbf{1.00} & \textbf{1.00} & \textbf{0.99}
& \textbf{1.00}
& 0.80 & 0.80 & 0.81 & 0.82 & \textbf{0.82}
& 0.70 & \textbf{0.90} \\

\bottomrule
\end{tabular}
}
\end{table*}

\begin{table*}[!ht]
\centering
\caption{
Detailed accuracy comparison on the second group of bivariate benchmark datasets
under dataset-specific tuning.
Baseline values are cited from \citet{QPE_2026causal}, while LRQS results are
obtained using dataset-specific hyperparameter tuning.
}
\label{tab:appendix_qpe_benchmarks_tuned_2}
\resizebox{\textwidth}{!}{
\begin{tabular}{lcccccccccccc}
\toprule
Method
& Multi & D4-s1 & D4-s2a & D4-s2b & D4-s2c
& Per & Sig & Vex & Qd-V & Sig-V & RbF-V & NN-V \\
\midrule

ANM
& 0.48 & 0.50 & 0.48 & 0.46 & 0.48
& 0.49 & 0.44 & 0.39 & 0.49 & 0.50 & 0.43 & 0.48 \\

DIVOT
& 0.34 & 0.50 & 0.57 & 0.55 & 0.55
& 0.97 & 0.82 & 0.05 & 0.32 & 0.44 & 0.63 & 0.47 \\

CVEL
& \textbf{0.97} & 0.67 & 0.51 & 0.58 & 0.58
& \textbf{1.00} & 0.84 & \textbf{0.96} & \textbf{0.91}
& \textbf{0.94} & 0.92 & 0.87 \\

QPE-k
& 0.88 & 0.58 & 0.67 & 0.61 & \textbf{0.64}
& 0.77 & 0.89 & 0.63 & 0.42 & 0.67 & 0.68 & 0.53 \\

QPE-f
& 0.96 & \textbf{0.79} & 0.71 & \textbf{0.62} & 0.60
& \textbf{1.00} & \textbf{0.90} & 0.91 & \textbf{0.91}
& 0.91 & \textbf{0.94} & \textbf{0.90} \\

\midrule

LRQS (est-g)
& 0.86 & 0.50 & 0.59 & 0.55 & 0.54
& 0.99 & 0.78 & 0.87 & 0.61 & 0.75 & 0.62 & 0.62 \\

LRQS (fix-g)
& 0.93 & 0.67 & \textbf{0.73} & 0.61 & 0.57
& 0.99 & 0.81 & 0.70 & 0.76 & 0.84 & 0.69 & 0.79 \\

\bottomrule
\end{tabular}
}
\end{table*}

\subsection{Sensitivity to fitted rank and effective spectral rank}
\label{appendix:rank_sensitivity}

\paragraph{Sensitivity to the fitted rank.}
We first examine the sensitivity of LRQS to the fitted number of
non-intercept basis functions, denoted by $K_{\mathrm{fit}}$.
We use the same 100 cause-effect pairs from the structural complexity
benchmark with Gaussian noise and $K_{\mathrm{gen}}=5$, and vary
$K_{\mathrm{fit}}\in\{2,3,4,5\}$ while keeping the other settings fixed
($n=1000$, $G=B=10$, $T_{\mathrm{in}}=5$, and
$T_{\mathrm{out}}=10$; for est-g, $m=5$).

In addition to causal-direction accuracy, we report the normalized
directional score gap. For
$s_{\mathrm{reverse}}+s_{\mathrm{forward}}>0$, define
\begin{equation}
\label{eq:directional-score-gap}
\delta
=
\frac{s_{\mathrm{reverse}}-s_{\mathrm{forward}}}
     {s_{\mathrm{reverse}}+s_{\mathrm{forward}}}.
\end{equation}
Positive values favor the true forward direction, negative values
favor the reverse direction, and values close to zero indicate
weak directional separation.
Table~\ref{tab:appendix_k_sensitivity} summarizes the results.

\begin{table}[!ht]
\centering
\caption{
Sensitivity of LRQS to the fitted rank $K_{\mathrm{fit}}$ on the same 100
structural complexity pairs with Gaussian noise and $K_{\mathrm{gen}}=5$.
The score gap is reported as median $[Q_1,Q_3]$.
``Decision flips'' denotes the fraction of pairs whose predicted direction
differs from that obtained with $K_{\mathrm{fit}}=2$.
}
\label{tab:appendix_k_sensitivity}
\begin{tabular}{lccccc}
\toprule
Method
& $K_{\mathrm{fit}}$
& Accuracy
& Median $\delta$ $[Q_1,Q_3]$
& Exact ties
& Decision flips \\
\midrule
fix-g
& 2 & 1.00 & 0.499 $[0.402,\,0.604]$ & 0 & -- \\
& 3 & 1.00 & 0.563 $[0.462,\,0.665]$ & 0 & 0\% \\
& 4 & 1.00 & 0.589 $[0.502,\,0.690]$ & 0 & 0\% \\
& 5 & 1.00 & 0.597 $[0.492,\,0.684]$ & 0 & 0\% \\
\midrule
est-g
& 2 & 1.00 & 0.543 $[0.423,\,0.665]$ & 0 & -- \\
& 3 & 1.00 & 0.616 $[0.469,\,0.737]$ & 0 & 0\% \\
& 4 & 0.99 & 0.652 $[0.463,\,0.822]$ & 0 & 1\% \\
& 5 & 0.79 & 0.871 $[0.466,\,1.000]^\dagger$ & 9 & 21\% \\
\bottomrule
\end{tabular}
\vspace{1mm}

{\footnotesize
$^\dagger$For est-g with $K_{\mathrm{fit}}=5$, the score-gap summary is
computed over the 91 non-tied pairs; $\delta$ is undefined for the nine exact
ties with $s_{\mathrm{forward}}=s_{\mathrm{reverse}}=0$.}
\end{table}

The fix-g variant is insensitive to the fitted rank over the evaluated range:
all 100 pairs are correctly oriented for every
$K_{\mathrm{fit}}\in\{2,3,4,5\}$, with no decision flips relative to
$K_{\mathrm{fit}}=2$.
The est-g variant is also stable for moderate changes in rank, retaining
99--100\% accuracy for $K_{\mathrm{fit}}\in\{2,3,4\}$.
However, increasing the fitted rank to $K_{\mathrm{fit}}=5$ reduces the
accuracy to 79\%, produces nine exact ties, and changes the predicted direction
for 21\% of the pairs relative to $K_{\mathrm{fit}}=2$.

These results show that increasing $K_{\mathrm{fit}}$ does not necessarily
improve causal identification.
In finite-sample estimation, $K_{\mathrm{fit}}$ is therefore better viewed as
a regularization parameter controlling the flexibility of the fitted quantile
surface than as an estimate of the number of generating components.
A conservatively small fitted rank preserves the directional asymmetry in this
experiment, whereas excessive flexibility can weaken causal discrimination.

\paragraph{Effective spectral rank.}
We next examine the effective rank of the empirical quantile surfaces in the
same structural complexity setting with Gaussian noise and
$K_{\mathrm{gen}}=5$.
For each of the 100 cause-effect pairs, we construct the $10\times10$
empirical quantile matrix in each direction and subtract its row means.
Let $\sigma_1\geq\sigma_2\geq\cdots$ denote the singular values of the
resulting row-centered matrix. We define the cumulative spectral energy
explained by the first $J$ components as
\[
E(J)
=
\frac{\sum_{j=1}^{J}\sigma_j^2}
     {\sum_j \sigma_j^2}.
\]
Because this benchmark uses $g=\mathrm{id}$, this analysis is performed
directly on the empirical quantile matrices before any rank projection and
does not use an estimated unwarping transformation.

\begin{figure}[t]
    \centering
    \includegraphics[width=\columnwidth]{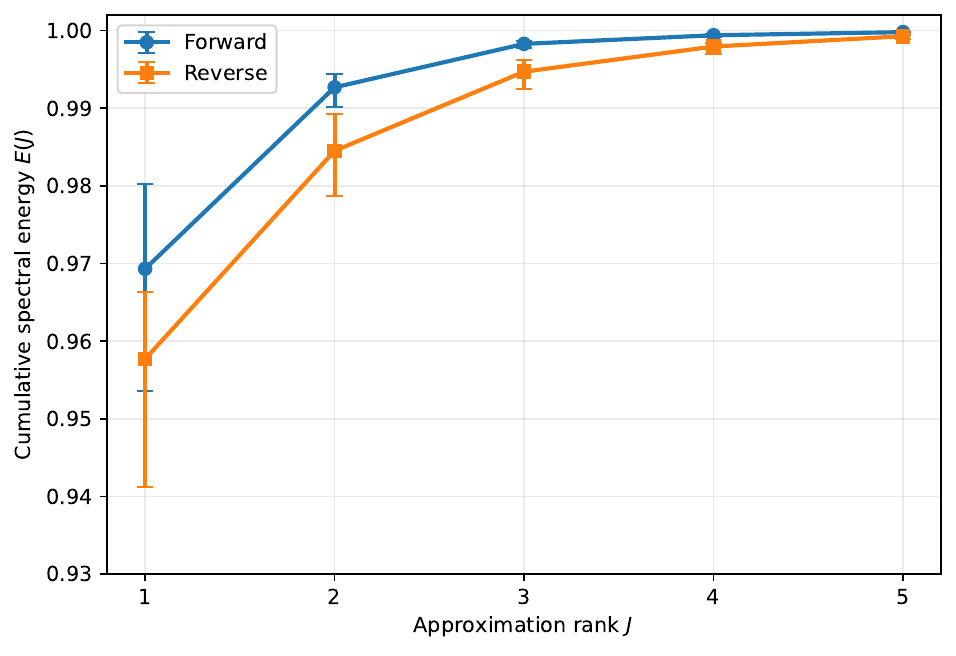}
    \caption{
    Cumulative spectral energy of the row-centered empirical quantile
    matrices for the structural complexity benchmark with Gaussian noise and
    $K_{\mathrm{gen}}=5$. Curves show the median across 100 cause-effect
    pairs, and error bars indicate the interquartile range.
    }
    \label{fig:effective_spectral_rank}
\end{figure}

As shown in Figure~\ref{fig:effective_spectral_rank}, the spectral energy is
strongly concentrated in the first few components.
At $J=2$, the median cumulative energy is $0.9927$
$[0.9902,0.9944]$ in the forward direction and $0.9845$
$[0.9786,0.9892]$ in the reverse direction.
The median paired difference
$\Delta(2)=E_{\mathrm{forward}}(2)-E_{\mathrm{reverse}}(2)$ is $0.0070$
$[0.0027,0.0139]$, with $\Delta(2)>0$ for 85\% of the pairs.

These results clarify that $K_{\mathrm{gen}}=5$ denotes the number of
components in the generating mechanism, rather than the effective rank of the
empirical quantile matrix. In this setting, the row-centered forward quantile matrix is empirically close to rank two, helping to explain why a small
$K_{\mathrm{fit}}$ remains effective under a five-component generating
mechanism. Both directions exhibit strong spectral concentration, with a
modest but consistent advantage in the forward direction.

\subsection{Sensitivity to iterations and initialization}
\label{appendix:optimization_sensitivity}

\paragraph{Sensitivity to outer iterations.}
Algorithm~\ref{alg:LowRank} uses a finite number of alternating
isotonic-regression, inverse-update, and rank-projection steps.
To examine its empirical sensitivity to the number of outer iterations, we
use one representative pair from the LSNM-tanh Gaussian setting and run the
est-g variant with 50 different initializations for
$T_{\mathrm{out}}\in\{0,1,5,10\}$.

\begin{figure}[t]
    \centering
    \includegraphics[width=\columnwidth]{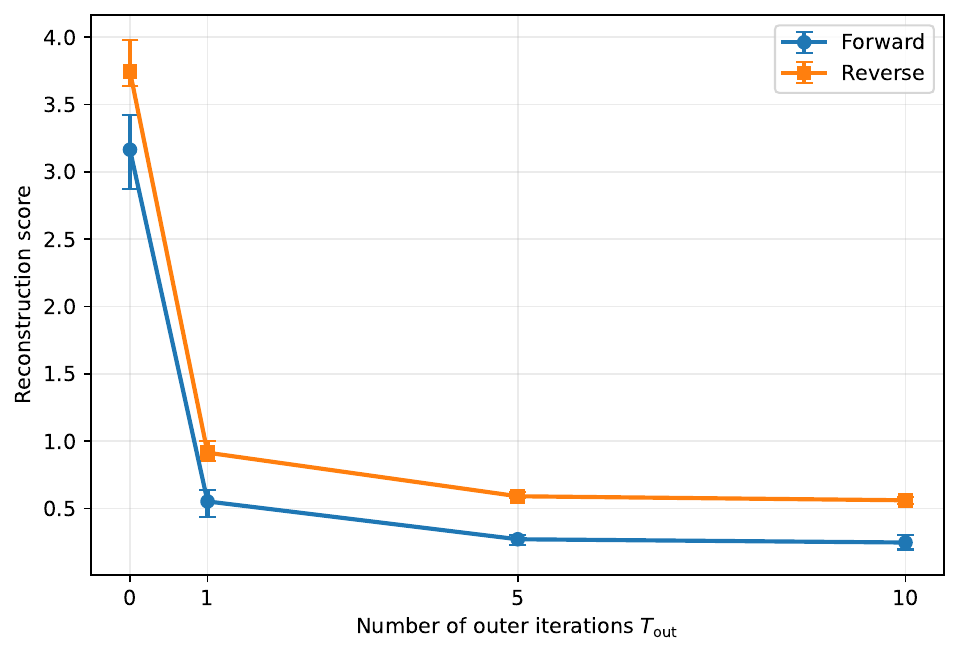}
    \caption{
    Sensitivity of LRQS (est-g) to the number of outer iterations
    $T_{\mathrm{out}}$ on one representative LSNM-tanh Gaussian pair.
    Curves show the median reconstruction score across 50 initializations,
    and error bars indicate the interquartile range.
    }
    \label{fig:e2_outer_iterations}
\end{figure}

As shown in Figure~\ref{fig:e2_outer_iterations}, both forward and reverse
reconstruction scores decrease substantially over the first few outer
iterations. The median forward score decreases from $3.165$ at
$T_{\mathrm{out}}=0$ to $0.271$ at $T_{\mathrm{out}}=5$, while the
corresponding reverse score decreases from $3.745$ to $0.590$.
The changes from $T_{\mathrm{out}}=5$ to $10$ are comparatively small.
The numbers of initializations selecting the correct direction are
$45/50$, $47/50$, $50/50$, and $50/50$ for
$T_{\mathrm{out}}=0,1,5,10$, respectively.

These results indicate that, in this representative example, the
reconstruction scores and directional decisions empirically stabilize after
a modest number of outer iterations, although
Algorithm~\ref{alg:LowRank} does not come with a general convergence
guarantee.

\paragraph{Sensitivity to the number of initializations.}
We next evaluate sensitivity to the number of initializations $m$ using
100 pairs from the same LSNM-tanh Gaussian setting. For each direction, we
retain the minimum reconstruction score over the first $m$ initializations,
as in Algorithm~\ref{alg:Bivariate_LRQS}.

\begin{table}[t]
\centering
\caption{
Sensitivity of LRQS (est-g) to the number of initializations $m$ on
100 LSNM-tanh Gaussian pairs.
}
\label{tab:appendix_initialization_sensitivity}
\begin{tabular}{ccc}
\toprule
$m$ & Accuracy & Exact ties \\
\midrule
1  & 0.94 & 0 \\
5  & 0.95 & 0 \\
10 & 0.96 & 0 \\
50 & 0.97 & 0 \\
\bottomrule
\end{tabular}
\end{table}

As shown in Table~\ref{tab:appendix_initialization_sensitivity}, increasing
$m$ yields only modest improvements in accuracy. In particular, increasing
the number of initializations tenfold from the default $m=5$ to $m=50$
improves accuracy by two percentage points, from $95\%$ to $97\%$.
Thus, in this experiment, the default initialization count provides a
reasonable balance between accuracy and additional computation.

\subsection{Sensitivity to sample size and quantile-surface discretization}
\label{appendix:discretization_sensitivity}

We examine the sensitivity of LRQS to the sample size and the discretization
used to construct the empirical conditional quantile surface.
All experiments in this subsection use the est-g variant on 100
LSNM-tanh Gaussian cause-effect pairs.

\paragraph{Sample size and number of bins.}
Let $u_{\min}$ denote the smallest quantile level used in the grid. We first vary the sample size
$n\in\{250,1000,4000\}$ and the number of conditioning-variable bins
$G\in\{5,10,20,30\}$, while fixing $B=10$ and $u_{\min}=0.05$.

\begin{figure}[t]
    \centering
    \includegraphics[width=\columnwidth]
    {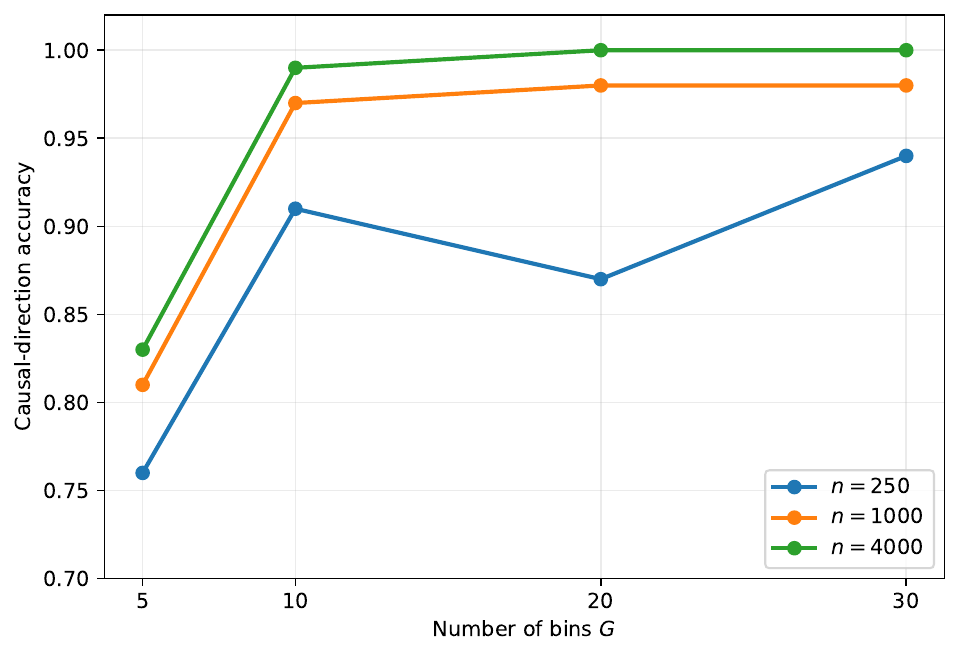}
    \caption{
    Sensitivity of LRQS (est-g) to the sample size $n$ and the number of
    conditioning-variable bins $G$ on 100 LSNM-tanh Gaussian pairs.
    The number of quantile levels is fixed to $B=10$ with
    $u_{\min}=0.05$. Each point reports causal-direction accuracy over
    the 100 pairs.
    }
    \label{fig:e3_sample_size_bin_sensitivity}
\end{figure}

As shown in Figure~\ref{fig:e3_sample_size_bin_sensitivity}, performance
degrades when the discretization along the conditioning variable is too
coarse. With $G=5$, accuracy ranges from $76\%$ to $83\%$ across the
evaluated sample sizes, with 3, 5, and 6 exact ties for
$n=250,1000,4000$, respectively.
In contrast, for $G\in\{10,20,30\}$ and $n\geq1000$, accuracy remains
between $97\%$ and $100\%$.
Thus, the default setting $n=1000$ and $G=10$ lies within a broader region
of high empirical accuracy, although overly coarse binning can degrade
directional discrimination.

\paragraph{Quantile-grid resolution and quantile range.}
We next examine the discretization along the quantile axis.
Under the default construction
$u_l=(2l-1)/(2B)$, changing $B$ also changes the outermost quantile level
$u_{\min}=1/(2B)$. We therefore interpret this experiment as a comparison of
quantile-grid resolutions rather than as an isolated effect of $B$.
With $n=1000$ and $G=10$, accuracy increases from $87\%$ at $B=5$
to $97\%$, $98\%$, and $100\%$ at $B=10$, $20$, and $30$,
respectively.

To separately examine sensitivity to the quantile range, we fix $G=B=10$
and vary $u_{\min}\in\{0.10,0.05,0.025,0.01\}$.
Table~\ref{tab:appendix_quantile_discretization} summarizes these analyses.

\begin{table}[t]
\centering
\caption{
Sensitivity of LRQS (est-g) to the quantile-grid resolution and quantile
range on LSNM-tanh Gaussian pairs. Each accuracy is computed over
100 cause-effect pairs.
}
\label{tab:appendix_quantile_discretization}

\begin{tabular}{cccc}
\toprule
\multicolumn{4}{c}{\textbf{Quantile-grid resolution}
($n=1000$, $G=10$)} \\
\midrule
$B$ & $u_{\min}$ & Accuracy & Exact ties \\
\midrule
5  & 0.1000 & 0.87 & 2 \\
10 & 0.0500 & 0.97 & 0 \\
20 & 0.0250 & 0.98 & 0 \\
30 & 0.0167 & 1.00 & 0 \\
\bottomrule
\end{tabular}

\vspace{2mm}

\begin{tabular}{ccccc}
\toprule
\multicolumn{5}{c}{\textbf{Quantile-range sensitivity}
($G=B=10$)} \\
\midrule
$n$
& $u_{\min}=0.10$
& $u_{\min}=0.05$
& $u_{\min}=0.025$
& $u_{\min}=0.01$ \\
\midrule
250  & 0.89 & 0.91 & 0.90 & 0.93 \\
1000 & 0.91 & 0.97 & 0.96 & 0.96 \\
\bottomrule
\end{tabular}
\end{table}

Across the explicitly varied quantile ranges, accuracy remains between
$89\%$ and $93\%$ for $n=250$ and between $91\%$ and $97\%$ for
$n=1000$, with no exact ties.
A joint higher-resolution setting with $n=1000$ and $G=B=30$ also achieves
$100\%$ accuracy.
Overall, these results indicate that performance is stable over a reasonable
neighborhood of the default discretization, while excessively coarse
binning or quantile grids can reduce accuracy in this benchmark.

\subsection{Directional score separation}
\label{appendix:directional_score_gap}

In a separate diagnostic experiment, we examine the normalized
directional score gap $\delta$ defined in
Eq.~(\ref{eq:directional-score-gap}) across the same 42 custom
benchmark settings used in Table~\ref{tab:custom_benchmarks},
with 100 cause-effect pairs per setting.

\begin{figure}[t]
    \centering
    \includegraphics[width=\columnwidth]
    {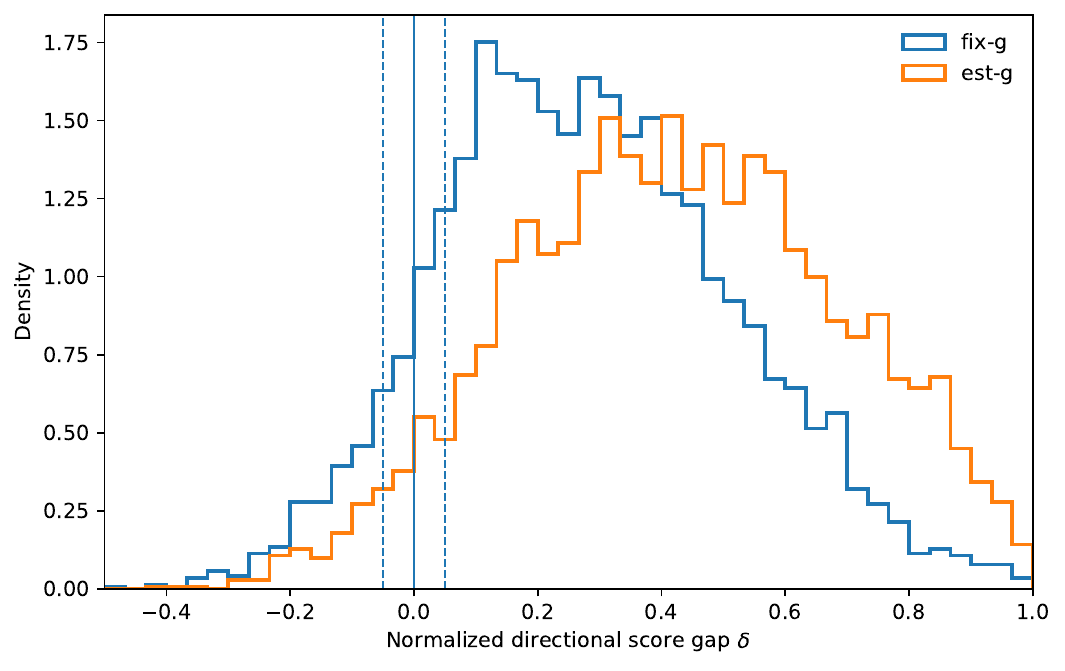}
    \caption{
    Distribution of the normalized directional score gap $\delta$ for LRQS
    (fix-g) and LRQS (est-g) across the 42 custom benchmark settings in a
    separate diagnostic experiment. Each setting contains 100 cause-effect
    pairs. The solid vertical line denotes $\delta=0$, and the dashed lines
    indicate the descriptive near-tie region $|\delta|<0.05$.
    }
    \label{fig:e4_directional_score_gap}
\end{figure}

Figure~\ref{fig:e4_directional_score_gap} shows that the score-gap
distributions are predominantly shifted toward positive values for both
variants. The median normalized gap is $0.270$ for fix-g and $0.422$ for
est-g. For descriptive purposes, we define a near tie as
$|\delta|<0.05$. Under this criterion, near ties account for $9.2\%$ of
the pairs for fix-g and $4.6\%$ for est-g, with no exact score ties observed
for either variant.

\begin{table}[t]
\centering
\caption{
Directional score separation for LRQS in the diagnostic experiment.
The near-tie threshold $|\delta|<0.05$ is used only as a descriptive
diagnostic.
}
\label{tab:appendix_directional_score_gap}
\small
\begin{tabular}{lcccc}
\toprule
Method
& Median $\delta$
& Near ties
& Median $|\delta|$ (correct)
& Median $|\delta|$ (incorrect) \\
\midrule
LRQS (fix-g)
& 0.270
& 9.2\%
& 0.302
& 0.081 \\
LRQS (est-g)
& 0.422
& 4.6\%
& 0.441
& 0.075 \\
\bottomrule
\end{tabular}
\end{table}

As shown in Table~\ref{tab:appendix_directional_score_gap}, correctly
identified pairs tend to exhibit substantially larger absolute score gaps
than incorrectly identified pairs. For fix-g, the median $|\delta|$ is
$0.302$ among correct decisions and $0.081$ among incorrect decisions;
for est-g, the corresponding values are $0.441$ and $0.075$.
Thus, incorrect decisions in this experiment tend to occur when the
forward and reverse reconstruction scores are relatively close, whereas
correct decisions generally exhibit clearer directional separation.

\subsection{Ablation of the low-rank constraint}
\label{appendix:low_rank_ablation}

To assess whether the directional discrimination of LRQS arises from the
low-rank constraint rather than from the flexibility of the monotone
unwarping alone, we perform an ablation using the est-g variant.
With $G=B=10$, the row-centered quantile matrix has rank at most $9$.
We therefore compare the standard setting $K_{\mathrm{fit}}=2$ with
$K_{\mathrm{fit}}=9$, which effectively removes the rank truncation.

We evaluate two representative settings from the custom benchmarks:
structural complexity with Gaussian noise and $K_{\mathrm{gen}}=5$, and
AN-tanh with Gaussian noise. Each setting contains the same 100 cause-effect
pairs used in the corresponding experiments.

\begin{table}[t]
\centering
\caption{
Ablation of the low-rank constraint for LRQS (est-g).
With $K_{\mathrm{fit}}=9$, the rank truncation is effectively removed for the
$10\times10$ empirical quantile matrices.
}
\label{tab:appendix_low_rank_ablation}
\begin{tabular}{lcc}
\toprule
Setting
& $K_{\mathrm{fit}}=2$
& $K_{\mathrm{fit}}=9$ \\
\midrule
Structural complexity
($K_{\mathrm{gen}}=5$, Gaussian)
& 100\% correct
& 100/100 exact ties \\
AN-tanh (Gaussian)
& 88\% correct
& 100/100 exact ties \\
\bottomrule
\end{tabular}
\end{table}

As shown in Table~\ref{tab:appendix_low_rank_ablation}, removing the
effective rank constraint eliminates directional discrimination in both
settings. For every pair with $K_{\mathrm{fit}}=9$, both directions are
perfectly reconstructed,
$s_{\mathrm{forward}}=s_{\mathrm{reverse}}=0$, so the causal direction
cannot be determined. In contrast, the standard low-rank setting retains
clear directional information.

This ablation shows that the performance of est-g in these experiments
cannot be attributed to the flexibility of the monotone unwarping alone.
Rather, the low-rank restriction provides an essential inductive bias for
preserving asymmetry between the two candidate directions.

\subsection{Stress test under monotonicity violations}
\label{appendix:monotonicity_violation}

The identifiability analysis assumes that the observation-level
transformation $g$ is strictly increasing. To examine empirical behavior
outside this assumed regime, we conduct a stress test using the non-monotone
transformation family
\[
g_{\kappa}(z)
=
\begin{cases}
z^4, & z \geq 0,\\
|z|^{4+\kappa}, & z < 0,
\end{cases}
\]
where $\kappa\in\{0,0.5,1,1.5,2\}$.

The transformation folds the latent variable around zero and is therefore
non-monotone for every value of $\kappa$. At $\kappa=0$, the transformation
is symmetric, whereas increasing $\kappa$ introduces greater asymmetry
between the negative and positive sides.

We evaluate this stress test on the structural complexity benchmark with
Gaussian noise, $K_{\mathrm{gen}}=5$, and $n=1000$, using 100 cause-effect
pairs. All baseline methods are evaluated on the same transformed versions of
these pairs, using the same method configurations as in the main comparison.
For LRQS, we set $G=B=10$ and vary the fitted rank over
$K_{\mathrm{fit}}\in\{1,2,3,4,5\}$. The fitted rank is not applicable to the
baseline methods.

\begin{table}[t]
\centering
\caption{
Causal-direction accuracy under violations of the increasing-transformation
assumption. Each entry reports accuracy over the same 100 transformed
cause-effect pairs. ``--'' indicates that $K_{\mathrm{fit}}$ is not applicable.
}
\label{tab:appendix_monotonicity_violation}
\small
\begin{tabular}{lcccccc}
\toprule
Method
& $K_{\mathrm{fit}}$
& $\kappa=0$
& $\kappa=0.5$
& $\kappa=1$
& $\kappa=1.5$
& $\kappa=2$ \\
\midrule
ANM             & -- & 0.26 & 0.25 & 0.32 & 0.24 & 0.26 \\
DIVOT           & -- & 0.00 & 0.00 & 0.00 & 0.00 & 0.00 \\
CVEL            & -- & 1.00 & 1.00 & 1.00 & 1.00 & 1.00 \\
QPE-k           & -- & 0.06 & 0.07 & 0.06 & 0.05 & 0.06 \\
QPE-f (fixed)   & -- & 0.00 & 0.00 & 0.00 & 0.00 & 0.00 \\
QPE-f (poly)    & -- & 0.00 & 0.00 & 0.00 & 0.00 & 0.00 \\
QPE-f (lowrank) & -- & 0.00 & 0.00 & 0.00 & 0.00 & 0.00 \\
\midrule
LRQS (est-g) & 1 & 1.00 & 1.00 & 1.00 & 1.00 & 1.00 \\
             & 2 & 1.00 & 1.00 & 1.00 & 1.00 & 1.00 \\
             & 3 & 1.00 & 1.00 & 1.00 & 1.00 & 1.00 \\
             & 4 & 1.00 & 1.00 & 1.00 & 1.00 & 1.00 \\
             & 5 & 0.92 & 0.93 & 0.95 & 0.90 & 0.86 \\
\midrule
LRQS (fix-g) & 1 & 1.00 & 1.00 & 1.00 & 1.00 & 1.00 \\
             & 2 & 1.00 & 1.00 & 1.00 & 1.00 & 1.00 \\
             & 3 & 1.00 & 1.00 & 1.00 & 1.00 & 1.00 \\
             & 4 & 1.00 & 1.00 & 1.00 & 1.00 & 1.00 \\
             & 5 & 1.00 & 1.00 & 1.00 & 1.00 & 1.00 \\
\bottomrule
\end{tabular}
\end{table}

As shown in Table~\ref{tab:appendix_monotonicity_violation}, CVEL achieves
perfect accuracy for every value of $\kappa$, whereas ANM obtains 24--32\%
and QPE-k obtains 5--7\%. DIVOT and all three QPE-f variants obtain 0\%
throughout this stress test. LRQS (fix-g) achieves perfect directional
accuracy across all evaluated values of $\kappa$ and $K_{\mathrm{fit}}$.
LRQS (est-g) also achieves perfect accuracy for $K_{\mathrm{fit}}\leq4$
across the entire transformation family. Its accuracy ranges from 86\% to
95\% for the more flexible $K_{\mathrm{fit}}=5$ setting.

Thus, robustness to this particular family of non-monotone observation
transformations varies substantially across methods. CVEL and LRQS remain
empirically effective throughout the stress test, and LRQS is most stable at
conservative fitted ranks.

\section{Computational resources and implementation details}\label{appendix:comp_resource_implementation_detail}

We provide the complete details of our computational environments, software
versions, and specific implementation patches applied to the baseline codes.

\subsection{Implementation details of LRQS}
\label{appendix:experiment_details}

\paragraph{Construction of the empirical quantile matrix.}
For a candidate direction $X \rightarrow Y$, we sort the observations by the
conditioning variable $X$ and partition the sorted indices into $G$ groups
using \texttt{numpy.array\_split}. This produces approximately equal-count bins
whose sample sizes are as equal as possible; when the sample size $n$ is not
divisible by $G$, the bin sizes differ by at most one.

Unless explicitly specified otherwise, the $B$ quantile levels are
\[
u_l = \frac{2l-1}{2B},
\qquad l=1,\ldots,B.
\]
For example, $B=10$ gives
$u_l \in \{0.05,0.15,\ldots,0.95\}$.
Within each bin, the corresponding empirical response quantiles are computed
using \texttt{numpy.quantile} without overriding its default quantile rule.

No special tie-aware binning rule is applied. Observations with identical
values of the conditioning variable are partitioned according to their
positions after sorting and are not explicitly constrained to remain in the
same bin. Ties in the response variable are passed directly to
\texttt{numpy.quantile}. For nonempty bins, no minimum-size rule, bin merging,
or additional smoothing is applied. If an empty bin occurs, its quantile row
is copied from the preceding bin; if the first bin is empty, its row is set to
zero. The same construction is applied in the reverse direction after
exchanging the conditioning and response variables.

\paragraph{Default settings.}
Unless otherwise stated, we use $G=10$ bins and $B=10$ quantile levels to
construct the empirical quantile surfaces, and set the number of non-intercept
basis functions to $K_{\mathrm{fit}}=2$.
In the iterative optimization process
(Algorithm~\ref{alg:LowRank}), the inner loop for rank-constrained projection
and the outer loop for alternating isotonic-regression and inverse-update steps
are repeated $T_{\mathrm{in}}=5$ and $T_{\mathrm{out}}=10$ times,
respectively.

For each additional initialization, we add
Gaussian perturbations
$\Sigma_{ij}\sim\mathcal{N}(0,0.5^2)$
to the initial surface to reduce sensitivity to initialization during the
alternating estimation procedure.

\subsection{Computational environments}

The experiments in this study were conducted across two distinct computing
environments depending on the hardware requirements of the evaluated methods.

\textbf{1. Local CPU environment:}
All experiments for our proposed LRQS method (both fix-g and est-g) and the
CPU-based baselines (ANM, DIVOT, QPE-k) were executed on a local workstation.

\begin{itemize}
    \item OS: Windows 11
    \item CPU: 11th Gen Intel(R) Core(TM) i7-1165G7 @ 2.80 GHz
    (4 cores, 8 threads)
    \item Memory (RAM): 8.0 GB (3200 MT/s)
\end{itemize}

\textbf{Software stack:}
Python 3.13.6. The core dependencies for executing our methods and local
baselines include NumPy (v2.4.2) and scikit-learn (v1.8.0).

\textbf{2. GPU computing cluster:}
The experiments involving the GPU-based baselines, CVEL and the three QPE-f
variants (fixed, poly, and lowrank), were conducted on a Docker computing
cluster to leverage hardware acceleration.

\begin{itemize}
    \item GPU: Single NVIDIA RTX 2080 Ti (11 GB VRAM)
\end{itemize}

\textbf{Software stack:}
We utilized the official NVIDIA PyTorch Docker container image
(\texttt{nvcr.io/nvidia/pytorch:23.12-py3}), which provides a highly
optimized environment containing Python 3.10 and PyTorch 2.2.0.

For the baseline methods, we used the official implementations
released by the original authors.
The evaluation protocol for Table~\ref{tab:main_results} is
described in Section~\ref{sec:experiments}.

\subsection{Implementation patches for QPE-f baseline}

To ensure a fair and reproducible evaluation of the three QPE-f variants
(fixed, poly, and lowrank), we utilized the official source code provided by~\citet{QPE_2026causal}. However, we encountered severe numerical
instability and compatibility issues when running the original implementation
in our GPU environment. Specifically, the model frequently produced
$-\infty$ for intermediate QPE scores, leading to
identical and trivial predictions across multiple benchmark datasets.

To conduct a valid comparative study, we introduced two minimal,
mathematically equivalent patches to the official codebase:

\begin{enumerate}
    \item \textbf{Type-hint compatibility fix:}
    We updated outdated type hints
    (\texttt{from torch.types import Tensor, Device})
    that caused import errors in newer PyTorch environments.

    \item \textbf{Numerical stabilization for Jacobian computation:}
    We traced the $-\infty$ issue to the denominator computation of the QPE
    term ($\partial u/\partial y$). In the original code, this was computed
    using the Jacobian-vector product (\texttt{jvp}), which frequently
    collapsed to zero in our environment, causing the term
    $\widehat{\mathrm{QPE}}
    =-(\partial u/\partial x)/(\partial u/\partial y)$
    to diverge. We patched this by replacing the \texttt{jvp}-based
    calculation with
    \texttt{log\_abs\_det\_jacobian(...).exp()},
    which retrieves the exact same analytical Jacobian through the
    normalizing flow's native method.
\end{enumerate}

This alternative implementation is mathematically equivalent, resolving the
$-\infty$ collapse and yielding dataset-specific performance scores
consistent with expectations.
\newpage
\IfFileExists{checklist.tex}{\section*{NeurIPS Paper Checklist}
\begin{enumerate}

\item {\bf Claims}
    \item[] Question: Do the main claims made in the abstract and introduction accurately reflect the paper's contributions and scope?
    \item[] Answer: \answerYes{} 
    \item[] Justification: the claims in the abstract and introductions match theoretical results and experimental results.
    \item[] Guidelines:
    \begin{itemize}
        \item The answer \answerNA{} means that the abstract and introduction do not include the claims made in the paper.
        \item The abstract and/or introduction should clearly state the claims made, including the contributions made in the paper and important assumptions and limitations. A \answerNo{} or \answerNA{} answer to this question will not be perceived well by the reviewers. 
        \item The claims made should match theoretical and experimental results, and reflect how much the results can be expected to generalize to other settings. 
        \item It is fine to include aspirational goals as motivation as long as it is clear that these goals are not attained by the paper. 
    \end{itemize}

\item {\bf Limitations}
    \item[] Question: Does the paper discuss the limitations of the work performed by the authors?
    \item[] Answer: \answerYes{} 
    \item[] Justification: We include a section at the end to discuss the limitations of our work.
    \item[] Guidelines:
    \begin{itemize}
        \item The answer \answerNA{} means that the paper has no limitation while the answer \answerNo{} means that the paper has limitations, but those are not discussed in the paper. 
        \item The authors are encouraged to create a separate ``Limitations'' section in their paper.
        \item The paper should point out any strong assumptions and how robust the results are to violations of these assumptions (e.g., independence assumptions, noiseless settings, model well-specification, asymptotic approximations only holding locally). The authors should reflect on how these assumptions might be violated in practice and what the implications would be.
        \item The authors should reflect on the scope of the claims made, e.g., if the approach was only tested on a few datasets or with a few runs. In general, empirical results often depend on implicit assumptions, which should be articulated.
        \item The authors should reflect on the factors that influence the performance of the approach. For example, a facial recognition algorithm may perform poorly when image resolution is low or images are taken in low lighting. Or a speech-to-text system might not be used reliably to provide closed captions for online lectures because it fails to handle technical jargon.
        \item The authors should discuss the computational efficiency of the proposed algorithms and how they scale with dataset size.
        \item If applicable, the authors should discuss possible limitations of their approach to address problems of privacy and fairness.
        \item While the authors might fear that complete honesty about limitations might be used by reviewers as grounds for rejection, a worse outcome might be that reviewers discover limitations that aren't acknowledged in the paper. The authors should use their best judgment and recognize that individual actions in favor of transparency play an important role in developing norms that preserve the integrity of the community. Reviewers will be specifically instructed to not penalize honesty concerning limitations.
    \end{itemize}

\item {\bf Theory assumptions and proofs}
    \item[] Question: For each theoretical result, does the paper provide the full set of assumptions and a complete (and correct) proof?
    \item[] Answer: \answerYes{} 
    \item[] Justification: each theoretical result is stated with all assumptions and is proved completely.
    \item[] Guidelines:
    \begin{itemize}
        \item The answer \answerNA{} means that the paper does not include theoretical results. 
        \item All the theorems, formulas, and proofs in the paper should be numbered and cross-referenced.
        \item All assumptions should be clearly stated or referenced in the statement of any theorems.
        \item The proofs can either appear in the main paper or the supplemental material, but if they appear in the supplemental material, the authors are encouraged to provide a short proof sketch to provide intuition. 
        \item Inversely, any informal proof provided in the core of the paper should be complemented by formal proofs provided in appendix or supplemental material.
        \item Theorems and Lemmas that the proof relies upon should be properly referenced. 
    \end{itemize}

    \item {\bf Experimental result reproducibility}
    \item[] Question: Does the paper fully disclose all the information needed to reproduce the main experimental results of the paper to the extent that it affects the main claims and/or conclusions of the paper (regardless of whether the code and data are provided or not)?
    \item[] Answer: \answerYes{} 
    \item[] Justification: the paper clearly states where its own numerical results are and where it used values reported in existing works. The code for the experiments is included.
    \item[] Guidelines:
    \begin{itemize}
        \item The answer \answerNA{} means that the paper does not include experiments.
        \item If the paper includes experiments, a \answerNo{} answer to this question will not be perceived well by the reviewers: Making the paper reproducible is important, regardless of whether the code and data are provided or not.
        \item If the contribution is a dataset and\slash or model, the authors should describe the steps taken to make their results reproducible or verifiable. 
        \item Depending on the contribution, reproducibility can be accomplished in various ways. For example, if the contribution is a novel architecture, describing the architecture fully might suffice, or if the contribution is a specific model and empirical evaluation, it may be necessary to either make it possible for others to replicate the model with the same dataset, or provide access to the model. In general. releasing code and data is often one good way to accomplish this, but reproducibility can also be provided via detailed instructions for how to replicate the results, access to a hosted model (e.g., in the case of a large language model), releasing of a model checkpoint, or other means that are appropriate to the research performed.
        \item While NeurIPS does not require releasing code, the conference does require all submissions to provide some reasonable avenue for reproducibility, which may depend on the nature of the contribution. For example
        \begin{enumerate}
            \item If the contribution is primarily a new algorithm, the paper should make it clear how to reproduce that algorithm.
            \item If the contribution is primarily a new model architecture, the paper should describe the architecture clearly and fully.
            \item If the contribution is a new model (e.g., a large language model), then there should either be a way to access this model for reproducing the results or a way to reproduce the model (e.g., with an open-source dataset or instructions for how to construct the dataset).
            \item We recognize that reproducibility may be tricky in some cases, in which case authors are welcome to describe the particular way they provide for reproducibility. In the case of closed-source models, it may be that access to the model is limited in some way (e.g., to registered users), but it should be possible for other researchers to have some path to reproducing or verifying the results.
        \end{enumerate}
    \end{itemize}

\item {\bf Open access to data and code}
    \item[] Question: Does the paper provide open access to the data and code, with sufficient instructions to faithfully reproduce the main experimental results, as described in supplemental material?
    \item[] Answer: \answerYes{} 
    \item[] Justification: The paper releases its code for experiments.
    \item[] Guidelines:
    \begin{itemize}
        \item The answer \answerNA{} means that paper does not include experiments requiring code.
        \item Please see the NeurIPS code and data submission guidelines (\url{https://neurips.cc/public/guides/CodeSubmissionPolicy}) for more details.
        \item While we encourage the release of code and data, we understand that this might not be possible, so \answerNo{} is an acceptable answer. Papers cannot be rejected simply for not including code, unless this is central to the contribution (e.g., for a new open-source benchmark).
        \item The instructions should contain the exact command and environment needed to run to reproduce the results. See the NeurIPS code and data submission guidelines (\url{https://neurips.cc/public/guides/CodeSubmissionPolicy}) for more details.
        \item The authors should provide instructions on data access and preparation, including how to access the raw data, preprocessed data, intermediate data, and generated data, etc.
        \item The authors should provide scripts to reproduce all experimental results for the new proposed method and baselines. If only a subset of experiments are reproducible, they should state which ones are omitted from the script and why.
        \item At submission time, to preserve anonymity, the authors should release anonymized versions (if applicable).
        \item Providing as much information as possible in supplemental material (appended to the paper) is recommended, but including URLs to data and code is permitted.
    \end{itemize}

\item {\bf Experimental setting/details}
    \item[] Question: Does the paper specify all the training and test details (e.g., data splits, hyperparameters, how they were chosen, type of optimizer) necessary to understand the results?
    \item[] Answer: \answerYes{} 
    \item[] Justification: All experiment setting and details are reported.
    \item[] Guidelines:
    \begin{itemize}
        \item The answer \answerNA{} means that the paper does not include experiments.
        \item The experimental setting should be presented in the core of the paper to a level of detail that is necessary to appreciate the results and make sense of them.
        \item The full details can be provided either with the code, in appendix, or as supplemental material.
    \end{itemize}

\item {\bf Experiment statistical significance}
    \item[] Question: Does the paper report error bars suitably and correctly defined or other appropriate information about the statistical significance of the experiments?
    \item[] Answer: \answerNo{} 
    \item[] Justification: The paper reports accuracy metrics over multiple runs.
    \item[] Guidelines:
    \begin{itemize}
        \item The answer \answerNA{} means that the paper does not include experiments.
        \item The authors should answer \answerYes{} if the results are accompanied by error bars, confidence intervals, or statistical significance tests, at least for the experiments that support the main claims of the paper.
        \item The factors of variability that the error bars are capturing should be clearly stated (for example, train/test split, initialization, random drawing of some parameter, or overall run with given experimental conditions).
        \item The method for calculating the error bars should be explained (closed form formula, call to a library function, bootstrap, etc.)
        \item The assumptions made should be given (e.g., Normally distributed errors).
        \item It should be clear whether the error bar is the standard deviation or the standard error of the mean.
        \item It is OK to report 1-sigma error bars, but one should state it. The authors should preferably report a 2-sigma error bar than state that they have a 96\% CI, if the hypothesis of Normality of errors is not verified.
        \item For asymmetric distributions, the authors should be careful not to show in tables or figures symmetric error bars that would yield results that are out of range (e.g., negative error rates).
        \item If error bars are reported in tables or plots, the authors should explain in the text how they were calculated and reference the corresponding figures or tables in the text.
    \end{itemize}

\item {\bf Experiments compute resources}
    \item[] Question: For each experiment, does the paper provide sufficient information on the computer resources (type of compute workers, memory, time of execution) needed to reproduce the experiments?
    \item[] Answer: \answerYes{} 
    \item[] Justification: The paper discusses computer resources of the experiments in the Appendix.
    \item[] Guidelines:
    \begin{itemize}
        \item The answer \answerNA{} means that the paper does not include experiments.
        \item The paper should indicate the type of compute workers CPU or GPU, internal cluster, or cloud provider, including relevant memory and storage.
        \item The paper should provide the amount of compute required for each of the individual experimental runs as well as estimate the total compute. 
        \item The paper should disclose whether the full research project required more compute than the experiments reported in the paper (e.g., preliminary or failed experiments that didn't make it into the paper). 
    \end{itemize}
    
\item {\bf Code of ethics}
    \item[] Question: Does the research conducted in the paper conform, in every respect, with the NeurIPS Code of Ethics \url{https://neurips.cc/public/EthicsGuidelines}?
    \item[] Answer: \answerYes{} 
    \item[] Justification: The research conform, in every respect, with the NeurIPS Code of Ethics.
    \item[] Guidelines:
    \begin{itemize}
        \item The answer \answerNA{} means that the authors have not reviewed the NeurIPS Code of Ethics.
        \item If the authors answer \answerNo, they should explain the special circumstances that require a deviation from the Code of Ethics.
        \item The authors should make sure to preserve anonymity (e.g., if there is a special consideration due to laws or regulations in their jurisdiction).
    \end{itemize}

\item {\bf Broader impacts}
    \item[] Question: Does the paper discuss both potential positive societal impacts and negative societal impacts of the work performed?
    \item[] Answer: \answerNA{} 
    \item[] Justification: The paper is a foundational research and has no directed path to any negative applications.
    \item[] Guidelines:
    \begin{itemize}
        \item The answer \answerNA{} means that there is no societal impact of the work performed.
        \item If the authors answer \answerNA{} or \answerNo, they should explain why their work has no societal impact or why the paper does not address societal impact.
        \item Examples of negative societal impacts include potential malicious or unintended uses (e.g., disinformation, generating fake profiles, surveillance), fairness considerations (e.g., deployment of technologies that could make decisions that unfairly impact specific groups), privacy considerations, and security considerations.
        \item The conference expects that many papers will be foundational research and not tied to particular applications, let alone deployments. However, if there is a direct path to any negative applications, the authors should point it out. For example, it is legitimate to point out that an improvement in the quality of generative models could be used to generate Deepfakes for disinformation. On the other hand, it is not needed to point out that a generic algorithm for optimizing neural networks could enable people to train models that generate Deepfakes faster.
        \item The authors should consider possible harms that could arise when the technology is being used as intended and functioning correctly, harms that could arise when the technology is being used as intended but gives incorrect results, and harms following from (intentional or unintentional) misuse of the technology.
        \item If there are negative societal impacts, the authors could also discuss possible mitigation strategies (e.g., gated release of models, providing defenses in addition to attacks, mechanisms for monitoring misuse, mechanisms to monitor how a system learns from feedback over time, improving the efficiency and accessibility of ML).
    \end{itemize}
    
\item {\bf Safeguards}
    \item[] Question: Does the paper describe safeguards that have been put in place for responsible release of data or models that have a high risk for misuse (e.g., pre-trained language models, image generators, or scraped datasets)?
    \item[] Answer: \answerNA{} 
    \item[] Justification: The paper does not propose or release data or models that have a high risk of misuse.
    \item[] Guidelines:
    \begin{itemize}
        \item The answer \answerNA{} means that the paper poses no such risks.
        \item Released models that have a high risk for misuse or dual-use should be released with necessary safeguards to allow for controlled use of the model, for example by requiring that users adhere to usage guidelines or restrictions to access the model or implementing safety filters. 
        \item Datasets that have been scraped from the Internet could pose safety risks. The authors should describe how they avoided releasing unsafe images.
        \item We recognize that providing effective safeguards is challenging, and many papers do not require this, but we encourage authors to take this into account and make a best faith effort.
    \end{itemize}

\item {\bf Licenses for existing assets}
    \item[] Question: Are the creators or original owners of assets (e.g., code, data, models), used in the paper, properly credited and are the license and terms of use explicitly mentioned and properly respected?
    \item[] Answer: \answerYes{} 
    \item[] Justification: All original assets are properly cited.
    \item[] Guidelines:
    \begin{itemize}
        \item The answer \answerNA{} means that the paper does not use existing assets.
        \item The authors should cite the original paper that produced the code package or dataset.
        \item The authors should state which version of the asset is used and, if possible, include a URL.
        \item The name of the license (e.g., CC-BY 4.0) should be included for each asset.
        \item For scraped data from a particular source (e.g., website), the copyright and terms of service of that source should be provided.
        \item If assets are released, the license, copyright information, and terms of use in the package should be provided. For popular datasets, \url{paperswithcode.com/datasets} has curated licenses for some datasets. Their licensing guide can help determine the license of a dataset.
        \item For existing datasets that are re-packaged, both the original license and the license of the derived asset (if it has changed) should be provided.
        \item If this information is not available online, the authors are encouraged to reach out to the asset's creators.
    \end{itemize}

\item {\bf New assets}
    \item[] Question: Are new assets introduced in the paper well documented and is the documentation provided alongside the assets?
    \item[] Answer: \answerYes{} 
    \item[] Justification: The paper releases code and explanations for running the code.
    \item[] Guidelines:
    \begin{itemize}
        \item The answer \answerNA{} means that the paper does not release new assets.
        \item Researchers should communicate the details of the dataset\slash code\slash model as part of their submissions via structured templates. This includes details about training, license, limitations, etc. 
        \item The paper should discuss whether and how consent was obtained from people whose asset is used.
        \item At submission time, remember to anonymize your assets (if applicable). You can either create an anonymized URL or include an anonymized zip file.
    \end{itemize}

\item {\bf Crowdsourcing and research with human subjects}
    \item[] Question: For crowdsourcing experiments and research with human subjects, does the paper include the full text of instructions given to participants and screenshots, if applicable, as well as details about compensation (if any)? 
    \item[] Answer: \answerNA{} 
    \item[] Justification: The paper does not involve crowdsourcing nor research with human subjects. 
    \item[] Guidelines:
    \begin{itemize}
        \item The answer \answerNA{} means that the paper does not involve crowdsourcing nor research with human subjects.
        \item Including this information in the supplemental material is fine, but if the main contribution of the paper involves human subjects, then as much detail as possible should be included in the main paper. 
        \item According to the NeurIPS Code of Ethics, workers involved in data collection, curation, or other labor should be paid at least the minimum wage in the country of the data collector. 
    \end{itemize}

\item {\bf Institutional review board (IRB) approvals or equivalent for research with human subjects}
    \item[] Question: Does the paper describe potential risks incurred by study participants, whether such risks were disclosed to the subjects, and whether Institutional Review Board (IRB) approvals (or an equivalent approval/review based on the requirements of your country or institution) were obtained?
    \item[] Answer: \answerNA{} 
    \item[] Justification: The paper does not involve crowdsourcing nor research with human subjects. 
    \item[] Guidelines:
    \begin{itemize}
        \item The answer \answerNA{} means that the paper does not involve crowdsourcing nor research with human subjects.
        \item Depending on the country in which research is conducted, IRB approval (or equivalent) may be required for any human subjects research. If you obtained IRB approval, you should clearly state this in the paper. 
        \item We recognize that the procedures for this may vary significantly between institutions and locations, and we expect authors to adhere to the NeurIPS Code of Ethics and the guidelines for their institution. 
        \item For initial submissions, do not include any information that would break anonymity (if applicable), such as the institution conducting the review.
    \end{itemize}

\item {\bf Declaration of LLM usage}
    \item[] Question: Does the paper describe the usage of LLMs if it is an important, original, or non-standard component of the core methods in this research? Note that if the LLM is used only for writing, editing, or formatting purposes and does \emph{not} impact the core methodology, scientific rigor, or originality of the research, declaration is not required.
    \item[] Answer: \answerNA{} 
    \item[] Justification: The methods proposed in this paper do not use LLMs as components.
    \item[] Guidelines:
    \begin{itemize}
        \item The answer \answerNA{} means that the core method development in this research does not involve LLMs as any important, original, or non-standard components.
        \item Please refer to our LLM policy in the NeurIPS handbook for what should or should not be described.
    \end{itemize}

\end{enumerate}}{}

\end{document}